\documentclass[preprint,12pt]{elsarticle}

\usepackage{amssymb}
\usepackage{amsmath}
\usepackage{amsfonts}
\usepackage{xcolor}
\usepackage{algorithm}
\usepackage{algpseudocode}
\usepackage{amsthm}
\usepackage{subcaption}

\usepackage{lineno}

\journal{Control Engineering Practice}

\newtheorem{theorem}{Theorem}[section]
\newtheorem{proposition}{Proposition}[section]
\newtheorem{lemma}{Lemma}[section]

\newtheorem{definition}{Definition}[section]
\newtheorem{assumption}{Assumption}[section]
\newtheorem{remark}{Remark}[section]

\begin{document}
\begin{frontmatter}


\title{Physics--guided neural networks for feedforward control \\with input--to--state--stability guarantees\tnoteref{label10}}
\tnotetext[label10]{This work is supported by the NWO research project PGN Mechatronics, project number 17973.}
\author[label1]{Max Bolderman\corref{cor1}}
\ead{m.bolderman@tue.nl}
\author[label1,labelASML]{Hans Butler}
\author[labelCanon]{Sjirk Koekebakker}
\author[labelPhilips]{Eelco van Horssen}
\author[labelASML]{Ramidin Kamidi}
\author[labelIBS]{Theresa Spaan--Burke}
\author[labelIBS]{Nard Strijbosch}
\author[label1]{Mircea Lazar}

\cortext[cor1]{Corresponding author.}
\affiliation[label1]{organization={Eindhoven University of Technology},addressline={Groene Loper 19},city={Eindhoven},postcode={5612 AP},country={The Netherlands}}
\affiliation[labelASML]{organization={ASML},
             addressline={De Run 6501},
             city={Veldhoven},
             postcode={5504 DR},
             country={The Netherlands}}
\affiliation[labelCanon]{organization={Canon Production Printing},
             addressline={St. Urbanusweg 43},
             city={Venlo},
             postcode={5900 MA},
             country={The Netherlands}}
\affiliation[labelPhilips]{organization={Philips Engineering Solutions},
             addressline={High Tech Campus 34},
             city={Eindhoven},
             postcode={5656 AE},
             country={The Netherlands}}
\affiliation[labelIBS]{organization={IBS Precision Engineering},
            addressline={ESP 201},
            city={Eindhoven},
            postcode={5633 AD},
            country={The Netherlands}}

\title{}

\begin{abstract}
The increasing demand on precision and throughput within high--precision mechatronics industries requires a new generation of feedforward controllers with higher accuracy than existing, physics--based feedforward controllers. As neural networks are universal approximators, they can in principle yield feedforward controllers with a higher accuracy, but suffer from bad extrapolation outside the training data set, which makes them unsafe for implementation in industry. Motivated by this, we develop a novel physics--guided neural network (PGNN) architecture that structurally merges a physics--based layer and a black--box neural layer in a single model. The parameters of the two layers are simultaneously identified, while a novel regularization cost function is used to prevent competition among layers and to preserve consistency of the physics--based parameters. Moreover, in order to ensure stability of PGNN feedforward controllers, we develop sufficient conditions for analyzing or imposing (during training) input--to--state stability of PGNNs, based on novel, less conservative Lipschitz bounds for neural networks. The developed PGNN feedforward control framework is validated on a real--life, high--precision industrial linear motor used in lithography machines, where it reaches a factor $2$ improvement with respect to physics--based mass--friction feedforward and it significantly outperforms alternative neural network based feedforward controllers.
\end{abstract}

\end{frontmatter}


\section{Introduction}
\label{sec:Introduction}
The field of high--precision mechatronics requires continuously innovating control methods to facilitate the ever--increasing demands on both throughput as well as accuracy. 
For example, wafer scanners in lithography machines used for semiconductor manufacturing~\cite{Heertjes2020} require sub--nanometer position accuracy at velocities and accelerations exceeding $1$ $\frac{m}{s}$ and $30$ $\frac{m}{s^2}$, respectively, see~\cite[Chapter~9]{Schmidt2014}. 
On a similar note, the ability to increase the throughput while decreasing the position error can allow for the use of components that are manufactured with larger tolerances and thereby improving the cost effectiveness.
This is an objective in, for example, the manufacturing industry of printing applications using relatively lower--cost stepping motors~\cite{Derammelaere2014}. 

Feedforward control is a dominant actor in achieving this high position accuracy, while feedback control predominantly concerns the closed--loop stability and disturbance rejection~\cite{Steinbuch2000}. 
Inverse model--based feedforward controllers generate the feedforward input by passing the reference through a model of the inverse system dynamics, and are therefore robust against varying references.
Conventionally, these models are derived from underlying physical knowledge, and can be linear~\cite{Boerlage2003, Dai2020}, linear in the parameters~\cite{Igarashi2021, Blanken2020b} or nonlinear~\cite{Jamaluding2009}. 
However, deriving a model from physical knowledge generally yields undermodelling which limits their performance when applied as feedforward controllers~\cite{Devasia2002}, e.g., the model does not include parasitic effects such as nonlinear friction and electromagnetic distortions that typically arise from manufacturing tolerances~\cite{Nguyen2015}. 
Therefore, more general model structures that can learn nonlinear and unknown parasitic effects are needed for improving performance of feedforward controllers.
Alternatively, iterative learning control (ILC) achieves superior performance with limited model accuracy, but it requires several repetitions of the same reference before this performance is reached~\cite{Bristow2006}.

Neural networks (NNs) are a good candidate for increasing the accuracy of feedforward controllers because of their universal approximation capabilities. Indeed, NNs  have already been used in system identification~\cite{Ljung2020} as well as to design feedforward controllers, see, e.g.,~\cite{Sorensen1999, Ren2009, Aarnoudse2021}. More recently, recurrent neural networks (RNNs) have also been used in identification, see e.g., \cite{Perrusquia2021},  \cite{Wang2017} (long--short--term mememory RNNs) and identification for feedforward control \cite{Hu2020} (gated recurrent unit RNNs). However, modeling the system dynamics as a black--box NN or RNN comes with the loss of underlying laws of physics, which increases sensitivity to the training data set. More specifically, as it will be shown in this paper, feedforward controllers based on black--box NNs extrapolate badly outside the training data set, which makes them unsafe for usage in high--precision mechatronics.

In order to enhance compliance of NNs with underlying physics laws, physics--guided or physics--informed neural networks (PINNs) were introduced in \citep{Karpatne2017} and \citep{Karniadakis2019}, respectively. In both these approaches, the system dynamics is still modeled as a black--box NN, but a loss function is used in training to penalize the deviation of the NN output from a physics--based model output. While the physics--based loss function promotes compliance with physics on the training data set, it does not necessarily improve extrapolation outside the training data set.

In this paper we develop a novel physics--guided neural network (PGNN) architecture for feedforward control that structurally merges a physics--based layer and a black--box neural layer in a single model. As the parameters of the two layers are simultaneously identified, this can lead to competition among layers and loss of interpretability of the physics--based parameters. To address this challenge, we develop a novel regularization cost function that prevents competition among layers and preserves interpretability of the physics--based parameters. Differently from black--box NNs or PINNs, which result in a NN as the feedforward controller, the developed PGNN feedforward controller consists of a physics--based and a NN--based layer, which delivers both higher precision and good extrapolation outside the training data set.

Moreover, we develop sufficient conditions for analyzing (after training) and imposing (during training) input--to--state stability of PGNN feedforward controllers based on a new type of Lipschitz bounds for neural networks, which is less conservative than existing bounds, see, e.g., \citep{Bonassi2021}. The developed conditions also provide a tight bound on the output of PGNN feedforward controllers, which is typically required in practice. The developed methodology is validated on a real--life industrial coreless linear motor from the lithography industry and compared with state--of--the--art alternative feedforward controllers, i.e., based on physics, black--box NNs and PINNs. The developed PGNN feedforward controller outperformes the physics--based feedforward controller by a factor $2$ in terms of the mean--absolute error and has significantly better accuracy and extrapolation outside the training data set with respect to the black--box NN and PINN feedforward controllers.


\begin{remark}
Compared to the authors' previous conference papers~\cite{Bolderman2021, Bolderman2022a}, we present the following original contributions in this journal paper: $i)$ generalized regularization cost function that includes regularization of parameters for both NN and physics--based layers and optimal selection of the regularization weights;
$ii)$ novel regularization cost based on simulated outputs of physics--based models, which explictely promotes PGNN (or PINN) complience with physics outside the training data set; $iii)$ ISS guarantees for PGNN feedforward controllers based on less conservative Lipschitz bounds, which enables the design of a stable nonlinear feedforward controller for nonminimum phase systems; $iv)$ novel, real--life experimental results.
\end{remark}

The remainder of this paper is organized as follows.
Sec.~\ref{sec:Preliminaries} introduces inversion--based feedforward control and states the considered research problem. 
Sec.~\ref{sec:PGNNFeedforward} introduces the novel PGNN architecture along with novel regularization cost functions, initialization and tuning methods. Conditions for analyzing and imposing ISS of PGNN feedforward controllers are developed in Sec.~\ref{sec:Stability}. 
Efficacy of the developed methodology is demonstrated on a real--life coreless linear motor (CLM) and a nonminimum phase simulation example in Sec.~\ref{sec:Validation}. Conclusions are summarized in Sec.~\ref{sec:Conclusions}. 

For streamlining exposition of the results, in this paper all proofs are reported in Appendices. 

\section{Preliminaries}
\label{sec:Preliminaries}
\subsection{Feedforward control preliminaries}
Fig.~\ref{fig:ControlScheme} displays a standard feedback--feedforward control scheme, where $u(t)$ is the control input and $y(t)$ the system output, with time $t \in \mathbb{R}_{\geq 0}$. 
The control objective in high--precision mechatronics is typically to minimize the tracking error $e(k) := r(k) - y(k)$, where $r(k)$ is the reference, and $y(k)$ is the measured output at discrete--time instants $k \in \mathbb{Z}_{\geq 0}$. 
\begin{figure}
    \centering
    \includegraphics[width=1.0\linewidth]{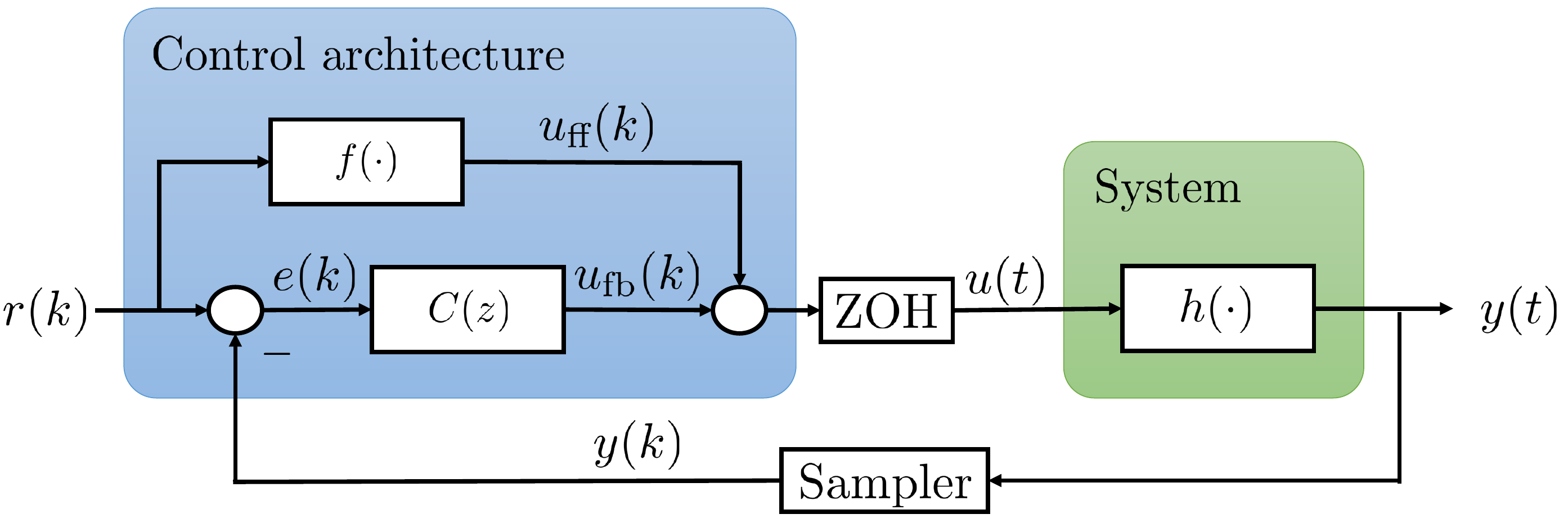}
    \caption{Feedback--feedforward control architecture.}
    \label{fig:ControlScheme}
\end{figure}
The control input $u(t)$ is computed at discrete--time instances $k$ according to
\begin{equation}
    \label{eq:InputDiscrete}
    u(k) = u_{\textup{fb}}(k) + u_{\textup{ff}}(k),
\end{equation}
where $u_{\textup{fb}}(k)$ is the feedback and $u_{\textup{ff}}(k)$ the feedforward input. The zero--order--hold (ZOH) in Fig.~\ref{fig:ControlScheme} is a discrete--to--continuous (D2c) operator, which lets $u(t) = u(k)$ for $t \in [k T_s, (k+1)T_s )$, with sampling time $T_s \in \mathbb{R}_{>0}$. 
The feedback input is given as 
\begin{equation}
    \label{eq:InputFeedback}
    u_{\textup{fb}}(k) = C(q) e(k),
\end{equation}
where $q$ is the forward shift operator, e.g., $e(k) = q e(k-1)$ and $C(q)$ a rational transfer function. 

We consider a nonlinear, discrete--time, input--output representation of the system dynamics, i.e., including the ZOH and the sampled outputs, such that
\begin{align}
\begin{split}
    \label{eq:NonlinearDynamics}
    y(k) = & h \big( [y(k-1), ..., y(k-n_a), u(k-n_k-1), ..., u(k-n_k-n_b)]^T \big),
\end{split}
\end{align}
where $n_a \in \mathbb{Z}_{\geq 0}$, $n_b \in \mathbb{Z}_{>0}$ are the order of the dynamics, $n_k \in \mathbb{Z}_{\geq 0}$ is the number of pure input delays, and $h : \mathbb{R}^{n_a+n_b} \rightarrow \mathbb{R}$ is a nonlinear function that describes the system dynamics. 
Assuming that there exists an exact inverse of $h$ in~\eqref{eq:NonlinearDynamics}, with a slight abuse of notation, we can define the ideal control input that inverts the system as 
\begin{align}
\begin{split}
    \label{eq:NonlinearDynamicsInverse}
    &u(k) =  h^{-1} \big( \phi(k) \big) \\
    &:=  h^{-1} \big( [y(k+n_k+1), ..., y(k+n_k-n_a+1), u(k-1), ..., u(k-n_b+1)]^T \big).
\end{split}
\end{align}
Consider that $y(k) = r(k)$, such that $e(k) = 0$ for all $k \in \mathbb{Z}_{\geq 0}$. 
Then, from~\eqref{eq:InputFeedback} we observe that $u_{\textup{fb}}(k) = 0$ which gives $u(k) = u_{\textup{ff}}(k)$ in~\eqref{eq:InputDiscrete}.
Consequently, substituting $u(k) = u_{\textup{ff}}(k)$ and $y(k) = r(k)$ in the inverse dynamics~\eqref{eq:NonlinearDynamicsInverse} yields the ideal feedforward controller
\begin{align}
\begin{split}
    \label{eq:NonlinearFeedforwardController}
    u_{\textup{ff}}(k) = & h^{-1} \big( \phi_{\textup{ff}}(k) \big) \\
    := & h^{-1} \big( [r(k+n_k+1), ..., r(k+n_k-n_a+1), \\
    & \quad \quad \quad u_{\textup{ff}}(k-1), ..., u_{\textup{ff}}(k-n_b+1)]^T \big). 
\end{split}
\end{align}
A linear feedforward controller is recovered by assuming that the system~\eqref{eq:NonlinearDynamics} is linear, such that
\begin{align}
\begin{split}
    \label{eq:LinearFeedforwardController}
    u_{\textup{ff}}(k)  & = \sum_{i=0}^{n_a} a_i r(k+n_k+1-i) - \sum_{i=1}^{n_b-1} b_i u_{\textup{ff}}(k-i) = G^{-1} (q) r(k),
\end{split}
\end{align}
where $a_i, b_i \in \mathbb{R}$ are the coefficients of the corresponding inverse transfer function $G^{-1}(q)$, and $G(q)$ is the transfer function of the linear system dynamics. In order to obtain an implementable feedforward controller~\eqref{eq:NonlinearFeedforwardController} or~\eqref{eq:LinearFeedforwardController}, we adopt the typical assumptions:

\begin{enumerate}
    \item \emph{Reference preview:} future reference values up to $r(k+n_k+1)$ are known at time instant $k$;
    \item \emph{Stable inverse dynamics:} the feedforward input $u_{\textup{ff}}(k)$ remains bounded for bounded reference $r(k)$. 
\end{enumerate}

\begin{remark}
    Stability of the linear feedforward~\eqref{eq:LinearFeedforwardController} is assessed by checking the poles of $G^{-1}(q)$. 
    When it is not stable, i.e., $G(q)$ is nonminimum phase, it is common practice to use non--causal filtering or stable approximate inversion techniques to obtain a minimum phase (approximation) of $G(q)$, see, e.g.,~\cite{Zundert2018} for an overview. 
\end{remark}
\begin{remark}
    In this paper we focus on the tracking problem, i.e., minimizing the output tracking error $e(k)$ with respect to a desired reference $r(k)$. 
    Nevertheless, the methods proposed in this work can be extended to reject known disturbances acting on the closed--loop system.
\end{remark}

In general, the function $h^{-1}$ in~\eqref{eq:NonlinearDynamics} is unknown.
Moreover, the precision demanded by industry exceeds manufacturing tolerances, which implies that a machine--specific $h^{-1}$ needs to be found. 
Even when designing a linear feedforward controller as in~\eqref{eq:LinearFeedforwardController}, the parameters $a_i$ and $b_i$ are, in principle, unknown. 
For these reasons, a systematic data--based approach for finding a model of the inverse system dynamics $h^{-1}$ would be desirable. 

\subsection{Identified inverse--model based feedforward control}
In order to introduce existing model classes for inversion--based feedforward control, we recall the three main ingredients required for direct identification of the inverse dynamics~\eqref{eq:NonlinearDynamicsInverse}, i.e., the data set, the model class, and the identification criterion.

\emph{\textbf{Data set:}} we consider the availability of a data set that is generated on the system displayed in Fig.~\ref{fig:ControlScheme}, i.e., satisfying the dynamics~\eqref{eq:NonlinearDynamics}.
As a result we have
\begin{equation}
    \label{eq:DataSet}
    Z^N = \{ \phi_0, u_0, \hdots, \phi_{N-1}, u_{N-1} \},
\end{equation}
where $\phi_i := \phi(i)$ and $u_i := u(i)$ for $i \in \{0, \hdots, N-1 \}$ in~\eqref{eq:NonlinearDynamicsInverse} during the data generating experiment, with $N \in \mathbb{Z}_{>0}$ the number of data points. 

\emph{\textbf{Model class:}} we require a model parametrization of the inverse system dynamics.
\begin{definition}
\label{def:ModelParametrization}
    A model parametrization of the inverse system dynamics is given as
    \begin{equation}
        \label{eq:ModelParametrization}
        \hat{u} \big( \theta, \phi(k) \big) = f \big( \theta, \phi(k) \big),
    \end{equation}
    where $\hat{u} \big( \theta, \phi(k) \big)$ is the prediction of the input $u (k)$, $\theta \in \mathbb{R}^{n_{\theta}}$ denotes the vector of parameters with $n_{\theta} \in \mathbb{Z}_{>0}$, and $f: \mathbb{R}^{n_{\theta}} \times \mathbb{R}^{n_a+n_b} \rightarrow \mathbb{R}$ is a user--defined function. 
\end{definition}
\emph{\textbf{Identification criterion:}} the identification criterion defines the best choice of parameters $\theta$ such that the output of the model~\eqref{eq:ModelParametrization} fits the output of the inverse system dynamics~\eqref{eq:NonlinearDynamicsInverse} on the data set~\eqref{eq:DataSet}. 
Typically, the identification criterion aims to minimize a cost function, i.e.,
\begin{equation}
    \label{eq:IdentificationCriterion}
    \hat{\theta} = \textup{arg} \min_{\theta} V ( \theta, Z^N),
\end{equation}
such as the mean--squared error (MSE) 
\begin{equation}
    \label{eq:CostFunctionMSE}
    V ( \theta, Z^N ) = V_{\textup{MSE}} (\theta, Z^N ) := \frac{1}{N} \sum_{i=0}^{N-1} \big(u_i - \hat{u}( \theta, \phi_i ) \big)^2.
\end{equation}
Similar to~\eqref{eq:NonlinearFeedforwardController} and~\eqref{eq:LinearFeedforwardController}, the feedforward controller is obtained by computing the input $u(k)$ that yields $y(k) = r(k)$ for the identified model, i.e.,~\eqref{eq:ModelParametrization} with $\theta = \hat{\theta}$ from~\eqref{eq:IdentificationCriterion}, such that
\begin{equation}
    \label{eq:FeedforwardIdentifiedGeneral}
    u_{\textup{ff}}(k) = \hat{u} \big( \hat{\theta} , \phi_{\textup{ff}}(k) \big) = f \big( \hat{\theta}, \phi_{\textup{ff}}(k) \big). 
\end{equation}

The model parametrization~\eqref{eq:ModelParametrization} is a crucial choice made by the user, since it characterizes the flexibility and robustness of the model. Two popular examples are:
\begin{enumerate}
    \item \emph{Physics--based} model (typically used by the high--precision mechatronics industry due to reasonable accuracy and good extrapolation \citep{Schmidt2014}), often derived from first--principle knowledge of the system, such that
    \begin{equation}
    \label{eq:PhysicsBasedParametrization}
    f \big( \theta, \phi(k) \big) = f_{\textup{phy}} \big( \theta_{\textup{phy}}, \phi(k) \big),
\end{equation}
where $\theta_{\textup{phy}} \in \mathbb{R}^{n_{\theta_{\textup{phy}}}}$ are the physical parameters.
    \item \emph{Black--box neural network} model (originally proposed for feedforward control in \citep{Sorensen1999}, but still currently not widely used in industrial practice despite higher accuracy, due to safety issues) as a universal approximator, such that
    \begin{align}
    \begin{split}
        \label{eq:NNParametrization}
        f \big( \theta,\phi(k) \big) & = f_{\textup{NN}} \big( \theta_{\textup{NN}}, \phi(k) \big) \\
        & = W_{L+1} \alpha_L \Big( \hdots \alpha_1 \big( W_1 \phi(k) + B_1  \big) \Big) + B_{L+1},
    \end{split}
    \end{align}
    with $\alpha_l : \mathbb{R}^{n_l} \rightarrow \mathbb{R}^{n_l}$ the aggregation of activation functions, $n_l \in \mathbb{Z}_{>0}$ the number of neurons in layer $l \in \{ 1, \hdots, L \}$, $L \in \mathbb{Z}_{>0}$ the number of hidden layers, and $\theta_{\textup{NN}} := [\textup{col}(W_1)^T, B_1^T, \hdots, \textup{col}(W_{L+1})^T, B_{L+1}^T]^T$ are all weights $W_l \in \mathbb{R}^{n_l\times n_{l-1}}$ and biases $B_l \in \mathbb{R}^{n_l}$, where $\textup{col}(W_l)$ stacks the columns of $W_l$. 
\end{enumerate}
\emph{Physics--informed neural networks} \citep{Karpatne2017, Karniadakis2019} also use a black--box NN model as in~\eqref{eq:NNParametrization}, but additionally employ a loss training cost function that penalizes the deviation of the NN model ouput from the physics--based model output on the data set. More specifically, the identification criterion~\eqref{eq:IdentificationCriterion} minimizes
\begin{equation}
    \label{eq:CostFunctionPINN}
        V (\theta, Z^N) = V_{\textup{MSE}} (\theta, Z^N) + c V_{\textup{phy}} (\theta, Z^N),
\end{equation}
where $c \in \mathbb{R}_{>0}$ defines the relative importance between data fit and the physical model compliance, which is given as
\begin{equation}
    \label{eq:PhysicalCompliance}
    V_{\textup{phy}} (\theta, Z^N) = \frac{1}{N} \sum_{i=0}^{N-1} \big( f_{\textup{phy}}(\theta_{\textup{phy}}, \phi_i ) - \hat{u} (\theta_{\textup{NN}}, \phi_i) \big)^2.
\end{equation}
During training of the PINN, it is possible to train both $\theta_{\textup{phy}}$ and $\theta_{\textup{NN}}$ simultaneously, or to train $\theta_{\textup{NN}}$ for a fixed $\theta_{\textup{phy}} = \theta_{\textup{phy}}^*$. Clearly, the PINN model class realizes a trade--off between data fit and compliance with the output of a physics--based model, which affects all the parameters of the NN, and hence, the accuracy of the resulting feedforward controller, as it will be shown in this paper.

\begin{remark} Alternatively to standard neural networks, i.e., feedforward neural networks \citep{Nelles2001}, recurrent neural networks (RNNs) can be used to model the inverse dynamics, as recently proposed in \cite{Hu2020, Perrusquia2021, Bonassi2022}. In general, RNNs have improved modeling capabilities, but training and implementing RNNs is more complex compared to feedforward NNs, which are just an input--ouput map. However, with respect to extrapolation properties, RNNs are still black--box models and hence, they do not necessarily behave well outside the training data set. As the high--precision mechatronics industry requires feedforward controllers with both high accuracy and good extrapolation outside the training data set, in this work we focus on merging physics--based models and feedforward NNs. A similar approach could be further adopted to merge physics--based models and RNNs, but this is beyond the scope of this paper.
\end{remark}

\subsection{Problem statement}
\label{sec:ProblemStatement}
To illustrate the intrinsic limitations of the physics--based~model \eqref{eq:PhysicsBasedParametrization} and of the NN--based model~\eqref{eq:NNParametrization}, we consider an industrial coreless linear motor as case study for feedforward control design, which is described in detail in Sec.~\ref{sec:ValidationRealLife}. For the sake of illustration, it suffices to state that the CLM is modelled as a moving mass experiencing nonlinear friction characteristics, such that Newton's second law gives the continuous time dynamics
\begin{equation}
    \label{eq:CLM_ContinuousTime}
    u(t) = m \ddot{y} (t) + F_{\textup{fric}} \big( y(t), \dot{y}(t) \big), 
\end{equation}
where $u(t)$ is the force input, $y(t)$ the position output, $\dot{y}(t)$ the velocity, $\ddot{y}(t)$ the acceleration, $m \in \mathbb{R}_{>0}$ the mass, and $F_{\textup{fric}}: \mathbb{R} \times \mathbb{R} \rightarrow \mathbb{R}$ the nonlinear friction.

\begin{figure}
    \centering
    \begin{subfigure}{0.49\linewidth}
        \includegraphics[width=1\linewidth]{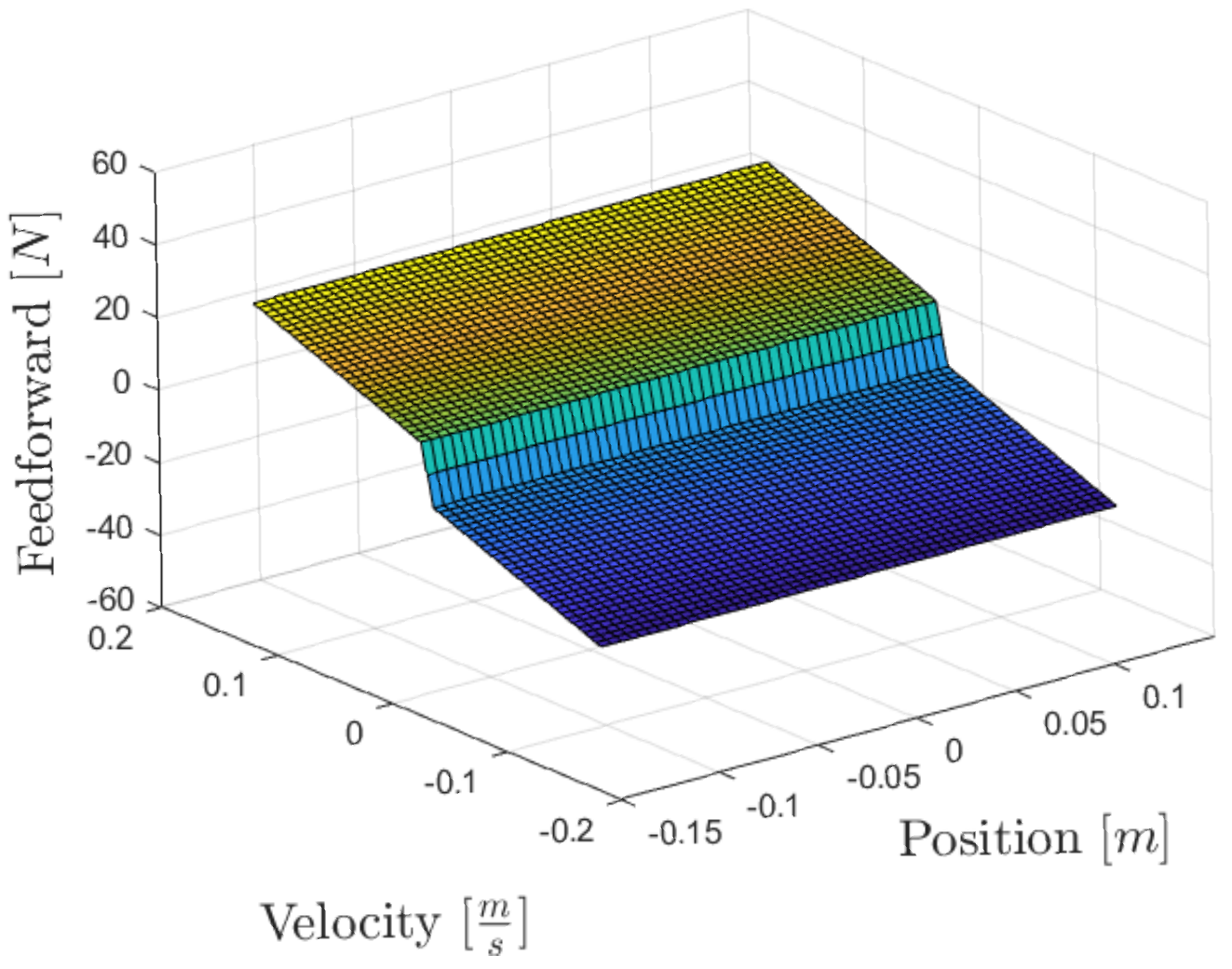}
        \caption{Physics--based inverse model.}
    \end{subfigure}
    \begin{subfigure}{0.49\linewidth}
        \includegraphics[width=1\linewidth]{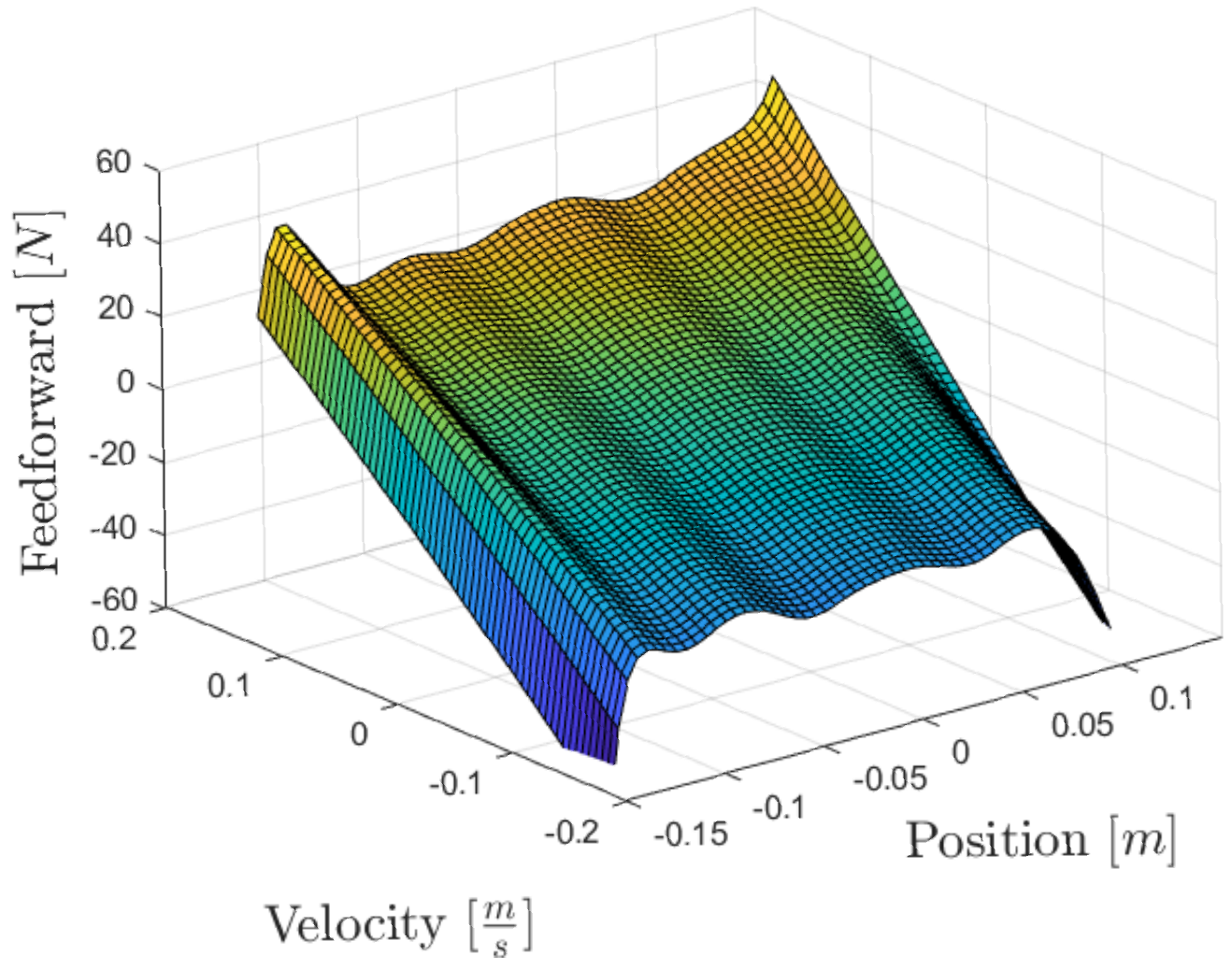}
        \caption{Black--box NN inverse model.}
    \end{subfigure}
    \caption{Identified friction $F_{\textup{fric}}$ for different model parametrizations obtained by computing the feedforward signal for varying reference positions $r(k)$ and velocities $\dot{r}(k)$ with zero acceleration $\ddot{r}(k) = 0$.}
    \label{fig:FrictionModels}
\end{figure}

Suppose that a reference $r(k)$ with zero acceleration is used, i.e., $\ddot{r}(k) = 0$. 
Then,~\eqref{eq:CLM_ContinuousTime} indicates that the feedforward controller only compensates for the friction $F_{\textup{fric}} \big( r(k), \dot{r}(k) \big)$. 
As a result, we visualize the friction model resulting from an identified physics--based model~\eqref{eq:PhysicsBasedParametrization} and a NN model~\eqref{eq:NNParametrization} in Fig.~\ref{fig:FrictionModels}, see Sec.~\ref{sec:ValidationRealLife} for details on the discretization, the models and the training.
Fig.~\ref{fig:FrictionModels} illustrates the main advantages and disadvantages of both approaches.

On one hand, the physics--based inverse model has limited accuracy, i.e., it does not learn the position dependency of the friction, but it extrapolates extremely well outside the training data set. Hence, it is safe to use in practice. The NN--based model on the other hand, identifies the position dependency of the friction, but it extrapolates very poorly outside the training data set. Hence it is not deemed safe for implementation in practice. 

\emph{Therefore, the problem considered in this paper is how to effectively merge physics--based models and neural networks within the context of feedforward control design, such that it becomes possible to obtain feedforward controllers that share the benefits of both types of models, i.e., high accuracy and good extrapolation.}

Note that existing approaches that combine physics--based models with neural networks do not offer a satisfactory solution to this problem. For example, a common approach \citep{Nelles2001} is to first identify a physics--based model and then to train the NN on the residuals. However, in this approach the NN cannot correct the bias of the physics--based model, due to sequential identification, which leads to a sub--optimal result and biased physics--based parameters when the NN cannot describe the residuals. Alternatively, PINNs \cite{Karpatne2017, Karniadakis2019} use a physics--based training cost function to promote compliance of a NN with a physics--based model on the training data set. However, since PINNs are essentially black--box NNs, good extrapolation outside the training data set is not necessarily guaranteed. 

To solve the stated problem, in the next section we introduce a novel physics--guided neural network architecture, which  merges a physics--based layer and a black--box neural network layer within a single model, and simultaneously identifies both the physics--based and the NN parameters. This brings two new challenges, to which we provide solutions in this paper: \emph{(i)}~how to avoid competition among the physics--based and NN layers and preserve consistency of the physics--based parameters (which is critical for good extrapolation) and \emph{(ii)}~how to guarantee robust stability of the resulting PGNN inverse model, which is a nonlinear model (hence, stability analysis methods for linear systems do not apply). \\




\section{Physics--guided neural networks for feedforward control}
\label{sec:PGNNFeedforward}
A schematic illustration of the developed PGNN architecture is shown in Fig.~\ref{fig:PGNN}, which we formally define as follows.
\begin{definition}
\label{def:PGNN}
The PGNN model parameterization of the inverse dynamics is given as
\begin{equation}
    \label{eq:PGNNGeneral}
    \hat{u} \big( \theta, \phi(k) \big) = f_{\textup{phy}} \big( \theta_{\textup{phy}} , \phi(k) \big) + f_{\textup{NN}} \big( \theta_{\textup{NN}}, T \big( \phi(k) \big) \big),
\end{equation}
where $\theta = [\theta_{\textup{NN}}^T, \theta_{\textup{phy}}^T]^T$ are the PGNN parameters which include the parameters of the NN--based layer $\theta_{\textup{NN}}$ and the parameters of the physics--based layer $\theta_{\textup{phy}}$. In addition, $T : \mathbb{R}^{n_a+n_b} \rightarrow \mathbb{R}^{n_0}$ is an input transformation, with $n_0 \in \mathbb{Z}_{>0}$ the number of NN inputs. 
\begin{figure}
    \centering
    \includegraphics[width=0.8\linewidth]{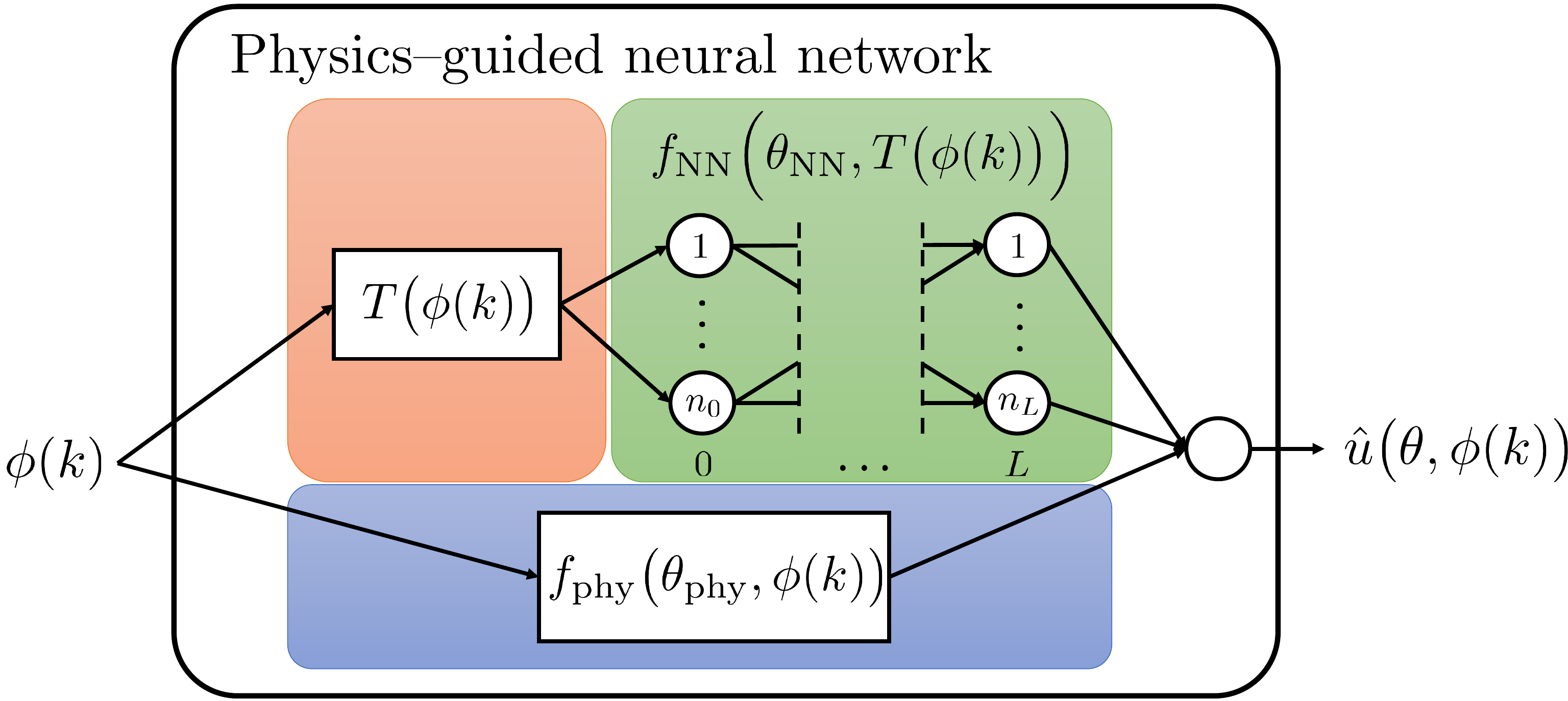}
    \caption{Physics--guided neural network architecture with physics and NN layers.}
    \label{fig:PGNN}
\end{figure}
\end{definition}
The input transformation $T$ can be used to improve numerical properties of the PGNN training, as well as to create physically relevant inputs, e.g., by converting consecutive output values to discrete velocities, or by imposing the rotational reproducible behaviour of rotary motors. 

\subsection{Regularized PGNN inverse system identification}
The PGNN~\eqref{eq:PGNNGeneral} is generally overparameterized, i.e., the NN can identify parts of the physical model which results in a parameter drift during training.
As a result, the physics--based layer does no longer constitute any physical interpretation, while the NN--based layer generates unnecessarily large outputs which amplify the undesired characteristics of the NN model.
Therefore, we employ a regularized cost function
\begin{equation}
    \label{eq:CostFunctionPGNN}
    V ( \theta, Z^N ) = V_{\textup{MSE}} (\theta, Z^N) + V_{\textup{reg}} (\theta),
\end{equation}
with the regularization cost given as
\begin{equation}
    \label{eq:RegularizationTerm}
    V_{\textup{reg}} (\theta) := \left\| \begin{bmatrix} \Lambda_{\textup{NN}} & 0 \\ 0 & \Lambda_{\textup{phy}} \end{bmatrix} \left( \theta - \begin{bmatrix}  0 \\ \theta_{\textup{phy}}^* \end{bmatrix} \right) \right\|_2^2.
\end{equation}
The parameters $\theta_\text{phy}^\ast$ are obtained as in \eqref{eq:IdentificationCriterion} with the inverse system dynamics parameterized by the physics--based model \eqref{eq:PhysicsBasedParametrization}. This is consistent with the so--called \emph{best linear approximation} that fits the data set, typically used in nonlinear system identification as initial parameter values \cite{Schoukens2019}. In~\eqref{eq:RegularizationTerm}, $\Lambda_{\textup{NN}} , \Lambda_{\textup{phy}}$ are matrices that define the relative importance of the different regularization terms. 
Note that, $\Lambda_{\textup{NN}}$ relates to the standard $\mathcal{L}_2$ regularization for the network weights and biases~\cite[Chapter~7]{Nelles2001}, while $\Lambda_{\textup{phy}}$ solves the overparameterization of the PGNN model~\eqref{eq:PGNNGeneral}. 
\begin{figure}
    \centering
    \includegraphics[width=0.9\linewidth]{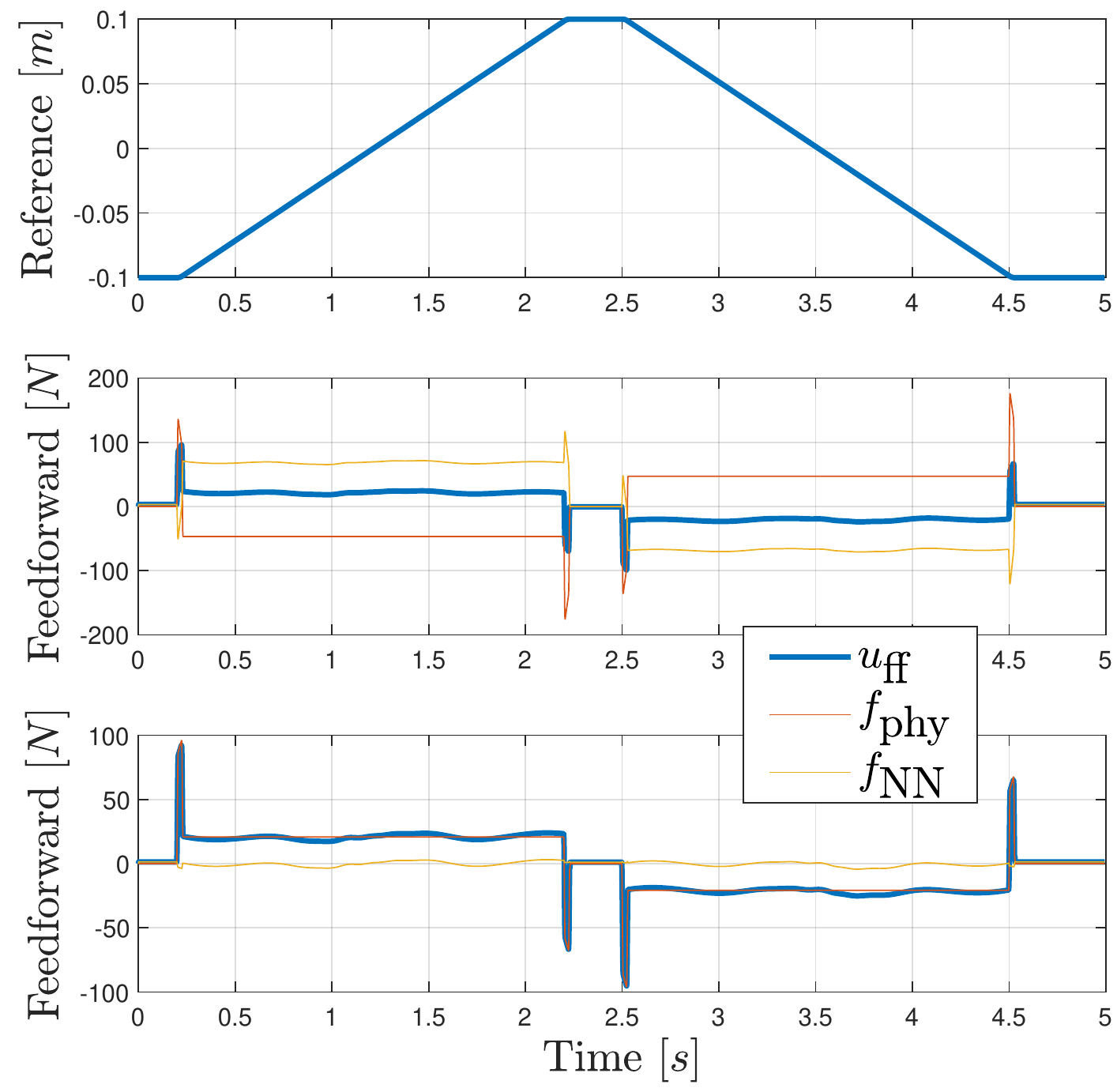}
    \caption{Reference signal (top window) and feedforward signal generated by the PGNN~\eqref{eq:PGNNGeneral} trained according to identification criterion~\eqref{eq:IdentificationCriterion} with cost function~\eqref{eq:CostFunctionMSE} (middle window), and cost function~\eqref{eq:CostFunctionPGNN} with $\Lambda_{\textup{phy}} = \textup{diag} (\theta_{\textup{phy}}^*)^{-1}$ and $\Lambda_{\textup{NN}} = 0$ (bottom window). }
    \label{fig:Feedforward_Regularization_Effect}
\end{figure}

Fig.~\ref{fig:Feedforward_Regularization_Effect} shows the feedforward signal generated by the physics--based and NN--based layer of the PGNN~\eqref{eq:PGNNGeneral}, which is trained using either the standard MSE cost function~\eqref{eq:CostFunctionMSE} (middle plot) or the regularized PGNN cost function~\eqref{eq:CostFunctionPGNN} (bottom plot) on data of the CLM, see Sec.~\ref{sec:ValidationRealLife} for details. For the regularized cost function~\eqref{eq:CostFunctionPGNN}, the NN--based layer augments the physics--based layer, while for the MSE cost function~\eqref{eq:CostFunctionMSE} layers start to compete and interpretability with physics is lost.

Next we provide some theoretical results regarding the properties of PGNN feedforward controllers. To this end we introduce the following formal definitions and assumptions. 
\begin{definition}
\label{def:UnknownDynamics}
    Given the physical model~\eqref{eq:PhysicsBasedParametrization}, we define the structural model error as 
    \begin{equation}
        \label{eq:StructuralModelError}
        g \big( \phi(k) \big) := h^{-1} \big( \phi(k) \big) - f_{\textup{phy}} \big( \theta_{\textup{phy}}^*, \phi(k) \big),
    \end{equation} 
    where $\theta_{\textup{phy}}^*$ are the physical parameters identified according to~\eqref{eq:IdentificationCriterion},~\eqref{eq:CostFunctionMSE}.
\end{definition}
Consequently, the inverse dynamics~\eqref{eq:NonlinearDynamicsInverse} can be rewritten into
\begin{equation}
    \label{eq:UnknownDynamics}
    u(k) = f_{\textup{phy}} \big( \theta_{\textup{phy}}^*, \phi(k) \big) + g \big( \phi(k) \big). 
\end{equation}
\begin{definition}
\label{def:OperatingConditions}
    Denote $\Phi_{\textup{ff}} \subseteq \mathbb{R}^{n_a+n_b}$ as all regressor points $\phi_{\textup{ff}}(k)$ supplied to the feedforward controller~\eqref{eq:PGNNGeneral} for all references $r(k)$ and all $k$. 
    Then, the set of operating conditions $\mathcal{R}$ is defined as
   \begin{equation}
       \label{eq:OperatingConditions}
       \mathcal{R} := \Phi_{\textup{ff}} \cup \{ \phi_0, ..., \phi_{N-1} \}.
    \end{equation}
\end{definition}
As an example, it is possible to obtain the operating conditions $\mathcal{R}$ by considering maxima and minima on the position, velocity, and acceleration. 

\begin{assumption}
\label{as:PGNNModelSet}
    There exists a $\theta_{\textup{NN}}^*$ such that $f_{\textup{NN}} \big( \theta_{\textup{NN}}^*, T\big(\phi(k) \big) \big) = g \big( \phi(k) \big)$ for all $\phi(k) \in \mathcal{R}$. 
\end{assumption}
Assumption~\ref{as:PGNNModelSet} dictates that there should exist a choice of parameters for which the PGNN~\eqref{eq:PGNNGeneral} recovers the original inverse system~\eqref{eq:NonlinearDynamicsInverse}, i.e., the system should be in the model class.
In practice, this is iteratively achieved by increasing the number of neurons $n_i$ or the number of layers $L$. 

\begin{assumption}
\label{as:PersistenceOfExcitation}
    For two sets of parameters $\theta_{\textup{NN}}^A \neq \theta_{\textup{NN}}^B$ with outputs $f_{\textup{NN}} \big( \theta_{\textup{NN}}^A, T\big(\phi(k)) \big) \neq f_{\textup{NN}} \big( \theta_{\textup{NN}}^B, T\big(\phi(k)) \big)$ for some $\phi(k) \in \mathcal{R}$, it holds that
    \begin{equation}
        \label{eq:PersistenceOfExcitation}
        \frac{1}{N} \sum_{i = 0}^{N-1} \left( f_{\textup{NN}} \big( \theta_{\textup{NN}}^A, T(\phi_i) \big) - f_{\textup{NN}} \big( \theta_{\textup{NN}}^B, T(\phi_i) \big) \right)^2 > 0.
    \end{equation}
\end{assumption}
Assumption~\ref{as:PersistenceOfExcitation} describes persistence of excitation in a nonlinear setting, i.e., if two sets of parameters give a different output in the operating conditions $\mathcal{R}$, this must be observed in the training data $Z^N$. 
Consequently, it is important to ensure that the training data reflects the operating conditions~\cite{Schoukens2019}.

\begin{assumption}
\label{as:TrainingConvergence}
    The minimization of~\eqref{eq:CostFunctionPGNN} over $\theta$ yields a global optimum. 
\end{assumption}
Since the cost function~\eqref{eq:CostFunctionPGNN} is in general non--convex, only local convergence guarantees can be established.  In practice, to avoid ending up in a local minimum, multiple trainings are performed with random parameter initialization.

\begin{proposition}[PGNN consistency]
\label{prop:ConsistentIdentification}
    Consider the PGNN~\eqref{eq:PGNNGeneral} that is used to identify the inverse dynamics~\eqref{eq:NonlinearDynamicsInverse} according to identification criterion~\eqref{eq:IdentificationCriterion} with cost function~\eqref{eq:CostFunctionPGNN} using $\Lambda_{\textup{NN}} = 0$ and $\Lambda_{\textup{phy}}$ full rank. 
    Suppose that Assumptions~\ref{as:PGNNModelSet},~\ref{as:PersistenceOfExcitation}, and~\ref{as:TrainingConvergence} hold.
    Then, the identified PGNN parameters satisfy $\hat{\theta} = [\hat{\theta}_{\textup{phy}}^T, \hat{\theta}_{\textup{NN}}^T]^T = [\theta_{\textup{phy}}^{*^T}, \theta_{\textup{NN}}^{*^T}]^T$. 
\end{proposition}
\begin{proof}
    See~\ref{app:propConsistentIdentification}. 
\end{proof}

    From Assumption~\ref{as:PGNNModelSet} and~\eqref{eq:UnknownDynamics}, it is observed that the identified PGNN recovers the inverse dynamics for all operating conditions, i.e.,
    \begin{equation}
    \label{eq:PGNNRecovers}
        f_{\textup{phy}} \big( \hat{\theta}_{\textup{phy}}, \phi(k) \big) + f_{\textup{NN}} \big( \hat{\theta}_{\textup{NN}}, \phi(k) \big) = h^{-1} \big( \phi(k) \big), \; \forall \; \phi(k) \in \mathcal{R}.
    \end{equation}

\begin{remark}
    A parallel linear--NN model structure was also employed for feedforward control design in~\cite{Kon2022}. 
    Therein, an alternative regularization method based on orthogonal projection was developed to avoid the competition between the linear and NN layers under the assumption that $g(\cdot)$ is identically zero outside $\mathcal{R}$. 
\end{remark}


\subsection{Optimized PGNN parameter selection}
Violation of Assumption~\ref{as:PGNNModelSet} or~\ref{as:TrainingConvergence} invalidates Proposition~\ref{prop:ConsistentIdentification}, since the system is not in the model class, or the optimization does not yield a global optimum.
In what follows, we derive a specific choice for the parameters $\theta$ which achieve a smaller value of the cost function~\eqref{eq:CostFunctionPGNN} compared to using the stand--alone physical model, i.e., we ensure the PGNN to improve over the physics.

We denote $\theta^{(j)}$ as the parameters at epoch $j \in \{ 0, ..., J\}$ during the optimization of~\eqref{eq:IdentificationCriterion}, with $J \in \mathbb{Z}_{\geq 0}$ the number of epochs before the solver terminates.
The identified parameter vector is then given as
\begin{equation}
    \label{eq:SolutionParameterEstimate}
    \hat{\theta} = \theta^{(j)}, \quad j = \textup{arg} \min_{j \in \{0, \hdots, J\}} V \big( \theta^{(j)}, Z^N \big). 
\end{equation}
For simplicity, we consider a linear--in--the--parameters (LIP) physical model~\eqref{eq:PhysicsBasedParametrization}, i.e., $f_{\textup{phy}} \big( \theta_{\textup{phy}}, \phi(k) \big) = \theta_{\textup{phy}}^T T_{\textup{phy}} \big( \phi(k) \big)$, and rewrite the PGNN~\eqref{eq:PGNNGeneral} according to
\begin{equation}
    \label{eq:PGNNLIP}
    \hat{u} \big( \theta, \phi(k) \big) = \theta_{\textup{L}}^T \phi_{\textup{L}} \big( \theta_{\textup{NL}}, \phi(k) \big) := \theta_{\textup{L}}^T \begin{bmatrix} \alpha_L \big( \theta_{\textup{NL}}, \phi(k) \big) \\ 1 \\ T_{\textup{phy}} \big( \phi(k) \big)  \end{bmatrix},
\end{equation}
where $\theta_{\textup{L}} := [\textup{col}(W_{L+1})^T, B_{L+1}, \theta_{\textup{phy}}^T]^T$ are the parameters in which the PGNN~\eqref{eq:PGNNLIP} is linear, and $\theta_{\textup{NL}} := [\textup{col}(W_1)^T, B_1^T, \hdots, \textup{col}(W_L)^T, B_L^T]^T$, such that $\theta = [\theta_{\textup{NL}}^T, \theta_{\textup{L}}^T]^T$.
From~\eqref{eq:PGNNLIP}, we observe that an equivalent of the physical model~\eqref{eq:PhysicsBasedParametrization} is obtained by selecting $\theta_{\textup{L}} = \overline{\theta}_{\textup{L}} := [0, 0, {\theta_{\textup{phy}}^{*^T}}]^T$. 
We rewrite the regularization term~\eqref{eq:RegularizationTerm} into
\begin{equation}
    \label{eq:RegularizationTermRewritten}
    V_{\textup{reg}} (\theta) = \left\| \begin{bmatrix} \Lambda_{\textup{NL}} & \Lambda \\ \Lambda^T & \Lambda_{\textup{L}} \end{bmatrix} \left( \begin{bmatrix} \theta_{\textup{NL}} \\ \theta_{\textup{L}} \end{bmatrix} - \begin{bmatrix} 0 \\ \overline{\theta}_{\textup{L}} \end{bmatrix} \right) \right\|_2^2,
\end{equation}
where $\Lambda$, $\Lambda_{\textup{NL}}^T = \Lambda_{\textup{NL}}$, and $\Lambda_{\textup{L}}^T = \Lambda_{\textup{L}}$ are obtained by selecting rows and columns of $\Lambda_{\textup{NN}}$ and $\Lambda_{\textup{phy}}$ accordingly. 
Then, given any parameter set $\theta_{\textup{NL}}^{(j)}$, choose $\theta_{\textup{L}}^{(j)}$ according to
\begin{align}
\begin{split}
    \label{eq:OptimizedInitialization}
    \theta_{\textup{L}}^{(j)} =&  M(\theta_{\textup{NL}}^{(j)})^{-1} \Bigg( \frac{1}{N} \sum_{i=0}^{N-1} u_i \phi_{\textup{L}} ( \theta_{\textup{NL}}^{(j)}, \phi_i )  \\ 
    & +  \Big( \big( \Lambda_{\textup{L}}^2 + \Lambda^T \Lambda \big)\overline{\theta}_{\textup{L}} - \big( \Lambda^T \Lambda_{\textup{NL}} + \Lambda_{\textup{L}} \Lambda^T \big) {\theta_{\textup{NL}}^{(j)}}  \Big) \Bigg),
\end{split}
\end{align}
where $M(\theta_{\textup{NL}}^{(j)})$ is given as
\begin{equation}
    \label{eq:OptimizedInitializationMatrix}
    M(\theta_{\textup{NL}}^{(j)}) := \frac{1}{N} \sum_{i=0}^{N-1} \phi_{\textup{L}} ( \theta_{\textup{NL}}^{(j)}, \phi_i ) \phi_{\textup{L}} ( \theta_{\textup{NL}}^{(j)}, \phi_i )^T + \big( \Lambda_{\textup{L}}^2 + \Lambda^T \Lambda \big).
\end{equation}
Note that~\eqref{eq:OptimizedInitialization} yields a unique solution when $M(\theta_{\textup{NL}}^{(j)})$ is nonsingular.
Consequently, nonsingularity of $M(\theta_{\textup{NL}}^{(j)})$ can be interpreted as persistence of excitation for the identification of $\theta_{\textup{L}}^{(j)}$. 

\begin{proposition}[PGNN improves over physics]
\label{prop:OptimizedInitialization}
    Consider the PGNN~\eqref{eq:PGNNLIP} with $\theta_{\textup{NL}}^{(j)}$ given, e.g., initialized randomly for $j=0$ or attained during training for $j \neq 0$. 
    Suppose that $M (\theta_{\textup{NL}}^ {(j)})$ is nonsingular, and 
    choose $\theta_{\textup{L}}^{(j)}$ according to~\eqref{eq:OptimizedInitialization}. 
    Then, for the cost function~\eqref{eq:CostFunctionPGNN}, we have
    \begin{equation}
        \label{eq:OptimizedInitializationResult1}
        V \left(\begin{bmatrix} \theta_{\textup{NL}}^{(j)} \\ \theta_{\textup{L}}^{(j)} \end{bmatrix}, Z^N \right) \leq V \left( \begin{bmatrix} \theta_{\textup{NL}}^{(j)} \\ \overline{\theta}_{\textup{L}} \end{bmatrix}, Z^N \right), 
    \end{equation}
    with strict inequality if and only if
    \begin{align}
    \begin{split}
        \label{eq:OptimizedInitializationCondition}
        \frac{1}{N} & \sum_{i=0}^{N-1} \left( u_i \phi_{\textup{L}} \big( \theta_{\textup{NL}}^{(j)}, \phi_i \big) - \phi_{\textup{L}} \big( \theta_{\textup{NL}}^{(j)}, \phi_i \big) \phi_{\textup{L}} \big( \theta_{\textup{NL}}^{(j)}, \phi_i \big)^T \overline{\theta}_{\textup{L}} \right)\\
        & \quad \quad \quad \quad -  ( \Lambda^T \Lambda_{\textup{NL}} + \Lambda_{\textup{L}} \Lambda^T )\theta_{\textup{NL}}^{(j)} \neq 0.
    \end{split}
    \end{align}
    Moreover, if $\Lambda_{\textup{NN}}$ is such that $\Lambda = 0$, i.e., the cross products between $\theta_{\textup{NL}}$ and $[W_{L+1}^T, B_{L+1}]^T$ are not regularized, the PGNN achieves a better data fit compared to the physical model, i.e., 
    \begin{equation}
        \label{eq:ImprovedMSE}
        V_{\textup{MSE}} \left( \begin{bmatrix} \theta_{\textup{NL}}^{(j)} \\ \theta_{\textup{L}}^{(j)} \end{bmatrix} , Z^N \right) \leq V_{\textup{MSE}} ( \theta_{\textup{phy}}^*, Z^N ),
    \end{equation}
    with strict inequality if~\eqref{eq:OptimizedInitializationCondition} holds.
\end{proposition}
\begin{proof}
    See~\ref{app:PropOptimizedInitialization}. 
\end{proof}


\begin{remark}
    Proposition~\ref{prop:OptimizedInitialization} can be directly extended for PGNNs with physical models that are not LIP. 
    To see this, rewrite the physical model 
    \begin{equation}
        \label{eq:PhysicalModelNonLIP}
        f_{\textup{phy}} \big( \theta_{\textup{phy}}, \phi(k) \big) = \theta_{\textup{L,phy}}^T T_{\textup{phy}} \big( \theta_{\textup{NL,phy}}, \phi(k) \big),
    \end{equation}
    and choose $\theta_{\textup{L}}^{(j)} = [\textup{col}(W_{L+1}^{(j)})^T, B_{L+1}^{(j)}, \theta_{\textup{L,phy}}^{{(j)}^T}]^T$ according to~\eqref{eq:OptimizedInitialization} using $\theta_{\textup{NL}}^{(j)} = [\textup{col}(W_1^{(j)})^{T}, B_1^{(j)^T}, ..., \textup{col}(W_L^{(j)})^T, B_L^{{(j)}^T}, {\theta_{\textup{NL,phy}}^{*}}^T]^{T}$.
\end{remark}

In practice, the optimized parameter selection is used after training, i.e., update $\theta_{\textup{L}}^{(j)}$ for~$j = J$, during training for each~$j=\{0, ..., J \}$, or as an initialization for~$j= 0$. 
Note that, from~\eqref{eq:SolutionParameterEstimate} and Proposition~\ref{prop:OptimizedInitialization}, if~\eqref{eq:OptimizedInitializationCondition} holds, we have that
\begin{equation}
\label{eq:UseageAsInitialization}
    V (\hat{\theta}, Z^N ) \leq V(\theta^{(0)}, Z^N) < V (\overline{\theta}, Z^N ),
\end{equation}
when~\eqref{eq:OptimizedInitialization} is used for initialization of $\theta_{\textup{L}}^{(0)}$.

\begin{figure}
    \centering
    \includegraphics[width=0.7\linewidth]{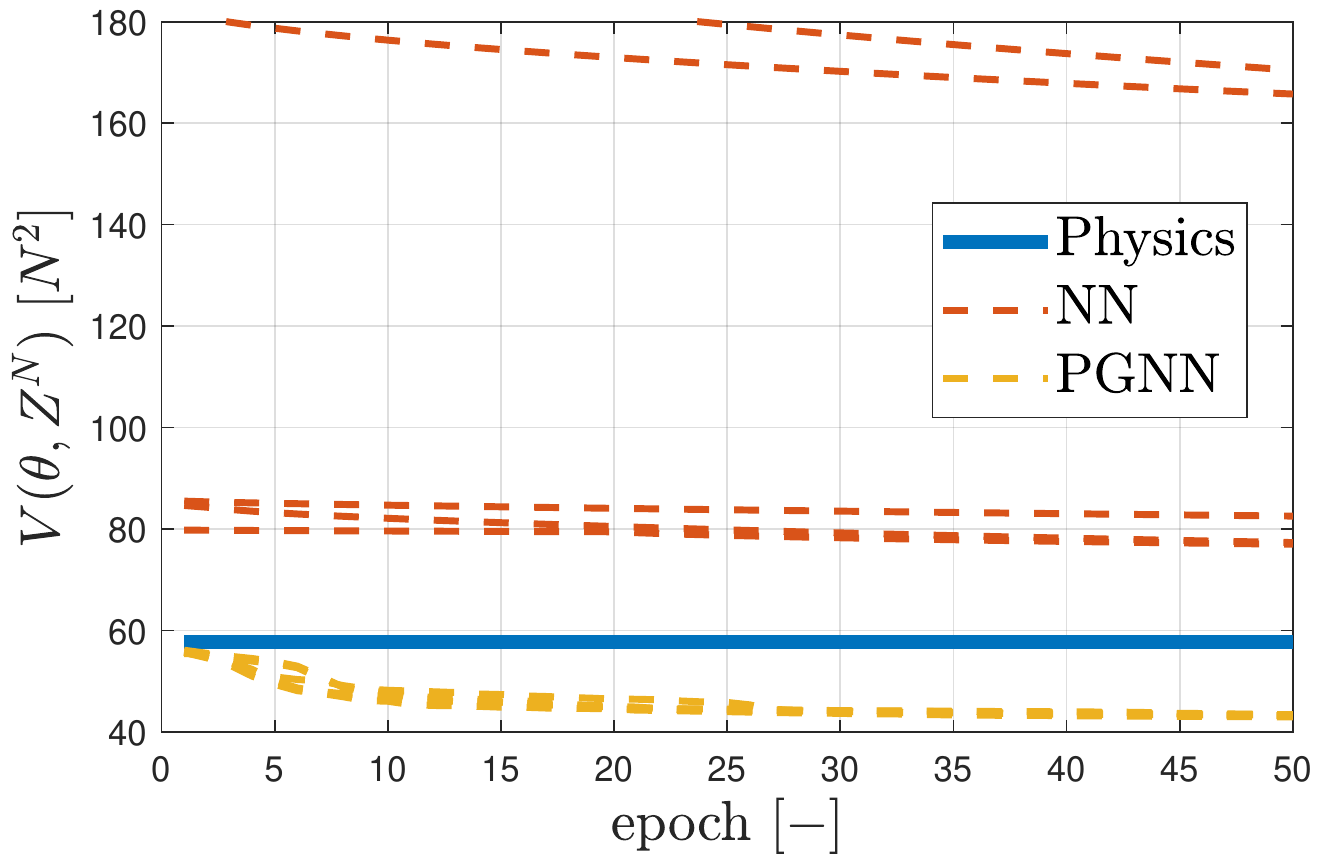}
    \caption{Training process of the NN~\eqref{eq:NNParametrization} and PGNN~\eqref{eq:PGNNGeneral} for $5$ independent trainings with random weight initialization and optimized initialization~\eqref{eq:OptimizedInitialization}. For comparison, $\frac{1}{N} \sum_{i=0}^{N-1} u_i^2 = 1477$~$N^2$.}
    \label{fig:Training_Convergence}
\end{figure}
We demonstrate the effectiveness of the optimized parameter initialization in Proposition~\ref{prop:OptimizedInitialization} by visualizing the value of the cost function for the first $50$ epochs for $5$ different trainings of the NN~\eqref{eq:NNParametrization} and the PGNN~\eqref{eq:PGNNGeneral} in Fig.~\ref{fig:Training_Convergence}, see Sec.~\ref{sec:ValidationRealLife} for the details.
Both the NN and PGNN are initialized with random $\theta_{\textup{NL}}^{(0)}$ and $\theta_{\textup{L}}^{(0)}$ according to~\eqref{eq:OptimizedInitialization}. 
The NN has difficulties to reach even the performance achieved by the stand--alone physics--based model, while the PGNN outperforms the physics--based model already from the first epoch.

\subsection{Enhancing PGNN extrapolation outside the training data set}
In general, there is no systematic method to validate Assumption~\ref{as:PersistenceOfExcitation} in practice. 
Even the guideline to sample the complete domain of interest in~\cite{Schoukens2019}, i.e., the operating conditions~$\mathcal{R}$, can be infeasible for reasons of time and safety.
As a result, we need to enhance robustness for situations in which the PGNN is operated on conditions that were not present in the training data.
Due to the lack of data, the physical model~\eqref{eq:PhysicsBasedParametrization} is the only source of reliable information for these conditions.
Consequently, we induce the robustness of the PGNN by promoting compliance with the physical model by means of regularization, such that~\eqref{eq:CostFunctionPGNN} becomes
\begin{equation}
    \label{eq:CostFunctionPGNNExtrapolation}
    V (\theta, Z^N) = V_{\textup{MSE}}(\theta, Z^N) + V_{\textup{reg}} (\theta) + \gamma V_{\textup{phy}} ( \theta, Z^E).
\end{equation}
In~\eqref{eq:CostFunctionPGNNExtrapolation}, $\gamma \in \mathbb{R}_{>0}$ is a regularization parameter, and the regularization cost is
\begin{equation}
    \label{eq:RegularizationTermPhysics}
    V_{\textup{phy}} (\theta, Z^E) := \frac{1}{E} \sum_{i = 0}^{E-1} \left( f_{\textup{phy}} ( \theta_{\textup{phy}}^* , \phi_i^E ) - \hat{u} ( \theta, \phi_i^E ) \right)^2.
\end{equation}
The set $Z^E = \{ \phi_0^E, \hdots, \phi_{E-1}^E \}$ describes the conditions for which we desire compliance of the PGNN with the physical model, i.e., the operating conditions $\mathcal{R}$ for which no data $Z^N$ is available.
Hence, we aim to have $Z^N \cup Z^E$ cover the operating conditions $\mathcal{R}$ up to high accuracy, which is performed automatically following Algorithm~\ref{alg:ZE}.
In Algorithm~\ref{alg:ZE}, $C ( \zeta, Z^N, Z^E)$ is an objective function that specifies the goal of the optimization. For example, choosing $\phi^E_i$ to maximize the minimum squared Euclidean distance with respect to all available regressor points, is achieved by choosing
\begin{equation}
    \label{eq:ZEObjectiveFunction}
    C(\zeta, Z^N, Z^E) := \min_{\phi \, \in Z^N, Z^E} \left\| \phi - \zeta \right\|_2^2.
\end{equation}
Algorithm~\ref{alg:ZE} iterates until a stopping criterion is met, e.g., by fixing a maximum number of points, or a minimum threshold $\epsilon \in \mathbb{R}_{>0}$ for the objective function~\eqref{eq:ZEObjectiveFunction}. 
\begin{algorithm}
\caption{Design algorithm for $Z^E$.}\label{alg:ZE}
\begin{algorithmic}
\State \textbf{Initialize} $Z^E = \{ \}$, $i = 0$,
\While{ $i < E-1$ $\wedge$ $C\big( \phi_{i-1}^E, Z^N, Z^E \big) > \epsilon$}
	\State $\phi_i^E =  \textup{arg} \max_{\zeta \in \mathcal{R}} C \big( \zeta , Z^N, Z^E \big)$,
	\State $Z^E = Z^E \cup \phi_i^E$,
	\State $i = i+1$.
\EndWhile
\end{algorithmic}
\end{algorithm}

\begin{remark}
    The results of Proposition~\ref{prop:OptimizedInitialization} can be extended directly to the cost function~\eqref{eq:CostFunctionPGNNExtrapolation} by appropriately revising the computation of $\theta_{\textup{L}}^{(j)}$ in~\eqref{eq:OptimizedInitialization}, i.e., compute the least squares solution of~\eqref{eq:CostFunctionPGNNExtrapolation} instead of~\eqref{eq:CostFunctionPGNN}. 
\end{remark}
\begin{figure}
    \centering
    \includegraphics[width=0.7\linewidth]{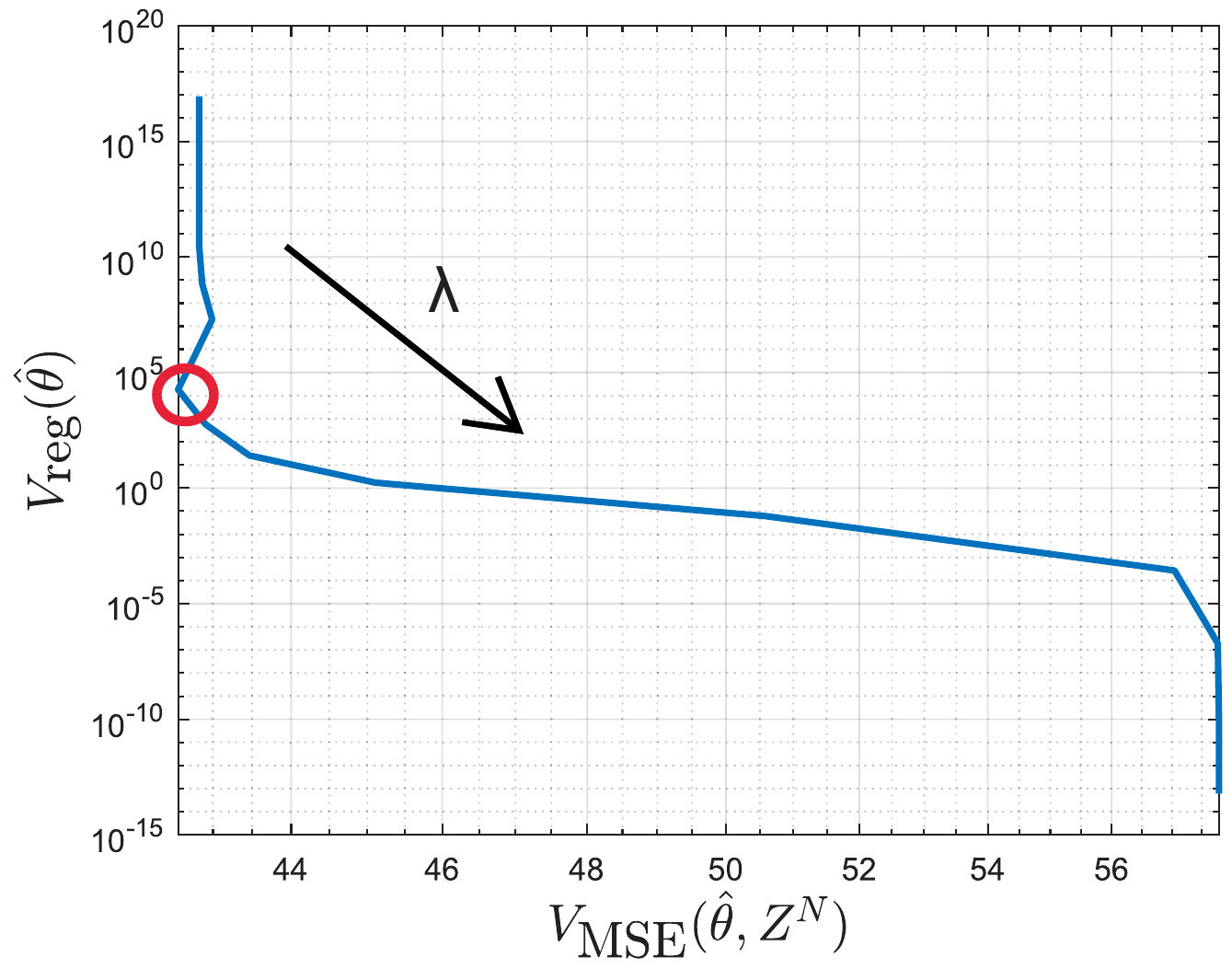}
    \caption{L--curve obtained by training the PGNN~\eqref{eq:PGNNGeneral} according to~\eqref{eq:IdentificationCriterion},~\eqref{eq:CostFunctionPGNN} with $\epsilon = 1$ in~\eqref{eq:RuleOfThumbLambda_phy} and $\Lambda_{\textup{NN}} = \lambda I$ with $20$ values of $\lambda$ logarithmically spaced in $[10^{-18}, 10^8]$, and optimal choice $\lambda = 10^{-5}$ (red circle). The regularization $V_{\textup{reg}}(\hat{\theta})$ is computed for $\Lambda_{\textup{NN}} = I$.}
    \label{fig:L_curve}
\end{figure}
Besides the choice of the NN dimensions, training the PGNN according to~\eqref{eq:CostFunctionPGNNExtrapolation} requires tuning of the hyperparameters $\Lambda_{\textup{phy}} \in \mathbb{R}^{n_{\theta_{\textup{phy}}} \times n_{\theta_{\textup{phy}}}}$, $\Lambda_{\textup{NN}} \in \mathbb{R}^{n_{\theta_{\textup{NN}}} \times n_{\theta_{\textup{NN}}}}$, and $\gamma \in \mathbb{R}_{>0}$.
The following rules--of--thumb are proposed for tuning these hyperparameters:
\begin{enumerate}
    \item Use $\Lambda_{\textup{phy}}$ to normalize for the magnitude of the physics--based parameters $\theta_{\textup{phy}}^*$, e.g.,
    \begin{equation}
        \label{eq:RuleOfThumbLambda_phy}
        \Lambda_{\textup{phy}} = \Big( \frac{1}{\epsilon \, n_{\textup{phy}}} \frac{1}{ N } \sum_{i=1}^{N} \big( u_i - f_{\textup{phy}} (\theta_{\textup{phy}}^*, \phi_i ) \big)^2 \Big)^{\frac{1}{2}} \textup{diag} (\theta_{\textup{phy}}^*)^{-1}.
    \end{equation}
    In~\eqref{eq:RuleOfThumbLambda_phy}, $\epsilon \in \mathbb{R}_{> 0}$ quantifies the relative deviation of the parameters in $\theta_{\textup{phy}}$ for which $V_{\textup{reg}}$ becomes $V_{\textup{MSE}}$ achieved by the physical model. 
    \item Choose $\Lambda_{\textup{NN}} = \lambda I$, where $\lambda \in \mathbb{R}_{\geq 0}$ can be tuned using, e.g., the L--curve~\cite{Hansen1993}, see Fig.~\ref{fig:L_curve}. Training the (PG)NN for a new value of $\lambda$ can be warm started with the trained parameters $\hat{\theta}$ resulting from the previous $\lambda$, which drastically reduces the computational burden. The parameter $\lambda$ can also be tuned using other, user--preferred, approaches, see, e.g.,~\cite{Bergstra2012}.
    \item $\gamma \in [0, 1]$ to quantify the relative importance of $Z^E$ with respect to $Z^N$. 
\end{enumerate}

\begin{figure}
    \centering
    \includegraphics[width=0.7\linewidth]{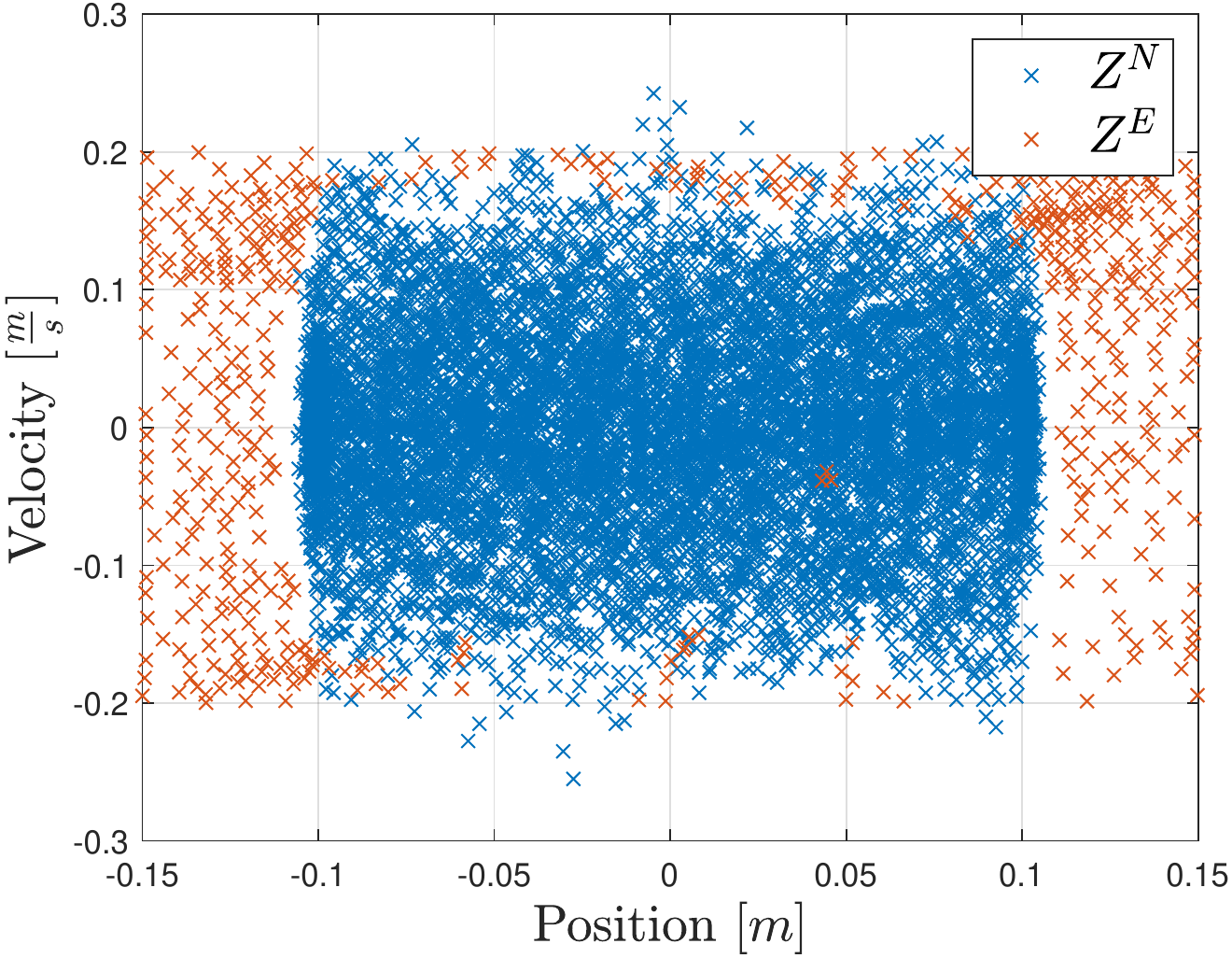}
    \caption{Two--dimensional illustration of $\phi_i$ in $Z^N$ generated on the CLM as discussed in Sec.~\ref{sec:Validation}, and the regressor points $\phi_i^E$ in $Z^E$ generated by Algorithm~\ref{alg:ZE}.}
    \label{fig:ExtrapolationDataSet}
\end{figure}
Let us consider again the CLM described by~\eqref{eq:CLM_ContinuousTime}, for which the operating conditions $\mathcal{R}$ are defined by a maximum on the position $| y(k) | < 0.15$~$m$ and velocity $| \delta y(k) | < 0.2$~$\frac{m}{s}$, with discrete--time differential operator $\delta = \frac{q-q^{-1}}{2T_s}$. 
We neglect the acceleration for the sake of simplicity.
The data set $Z^N$ does not cover the full range of $\mathcal{R}$, see Fig.~\ref{fig:ExtrapolationDataSet}, since it was deemed unsafe to travel the full stroke during the data generating experiment which included a dithering signal on the input. 
Consequently, application of Algorithm~\ref{alg:ZE} generates the set $Z^E$, see the orange crosses in Fig.~\ref{fig:ExtrapolationDataSet}, such that $Z^N \cup Z^E$ covers the operating conditions $\mathcal{R}$.

\begin{figure}
    \centering
    \begin{subfigure}{0.49\linewidth}
        \includegraphics[width=1\linewidth]{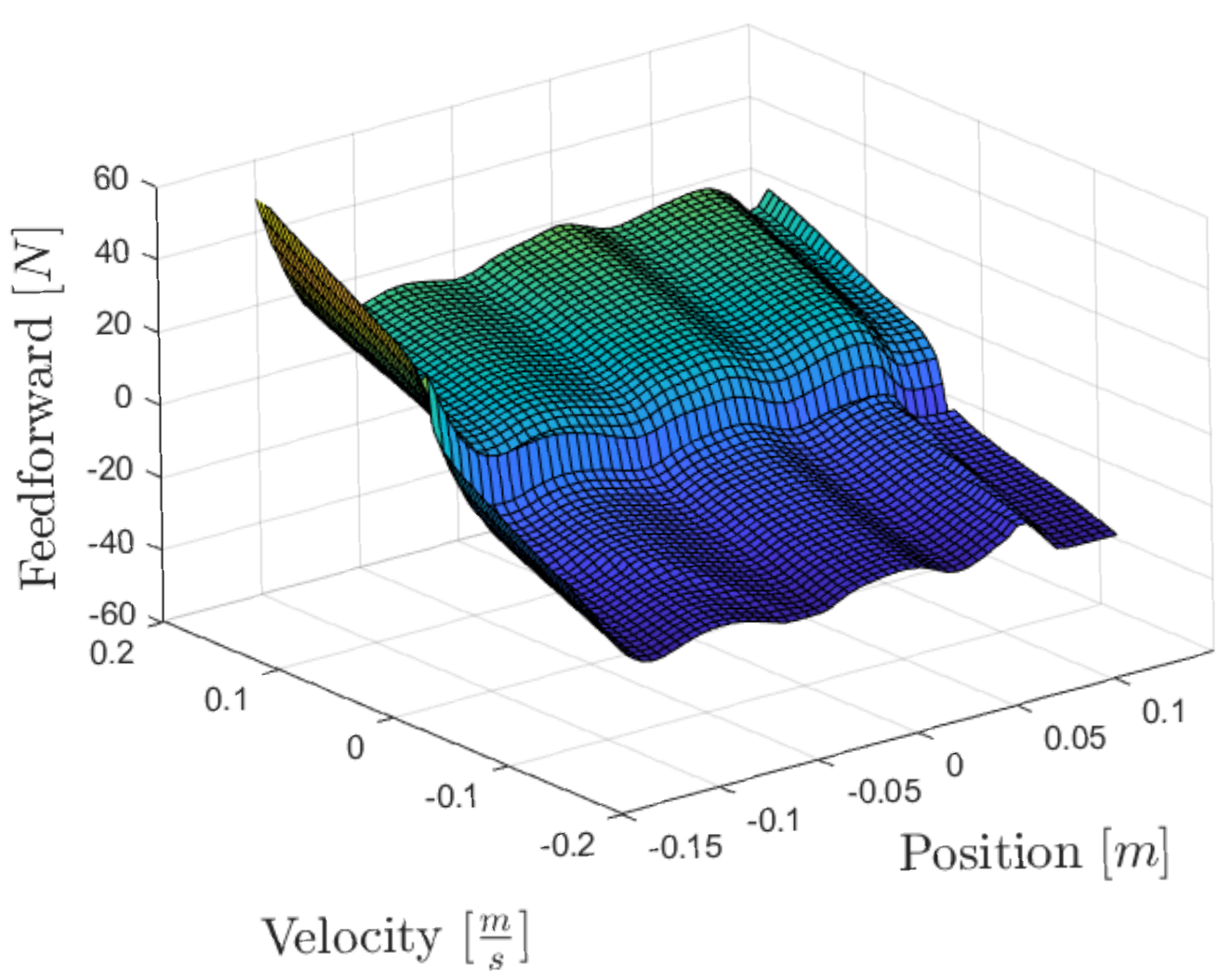}
        \caption{PGNN model with $\gamma = 0$.}
        \label{fig:FrictionModels_PGNN_A}
    \end{subfigure}
    \begin{subfigure}{0.49\linewidth}
        \includegraphics[width=1\linewidth]{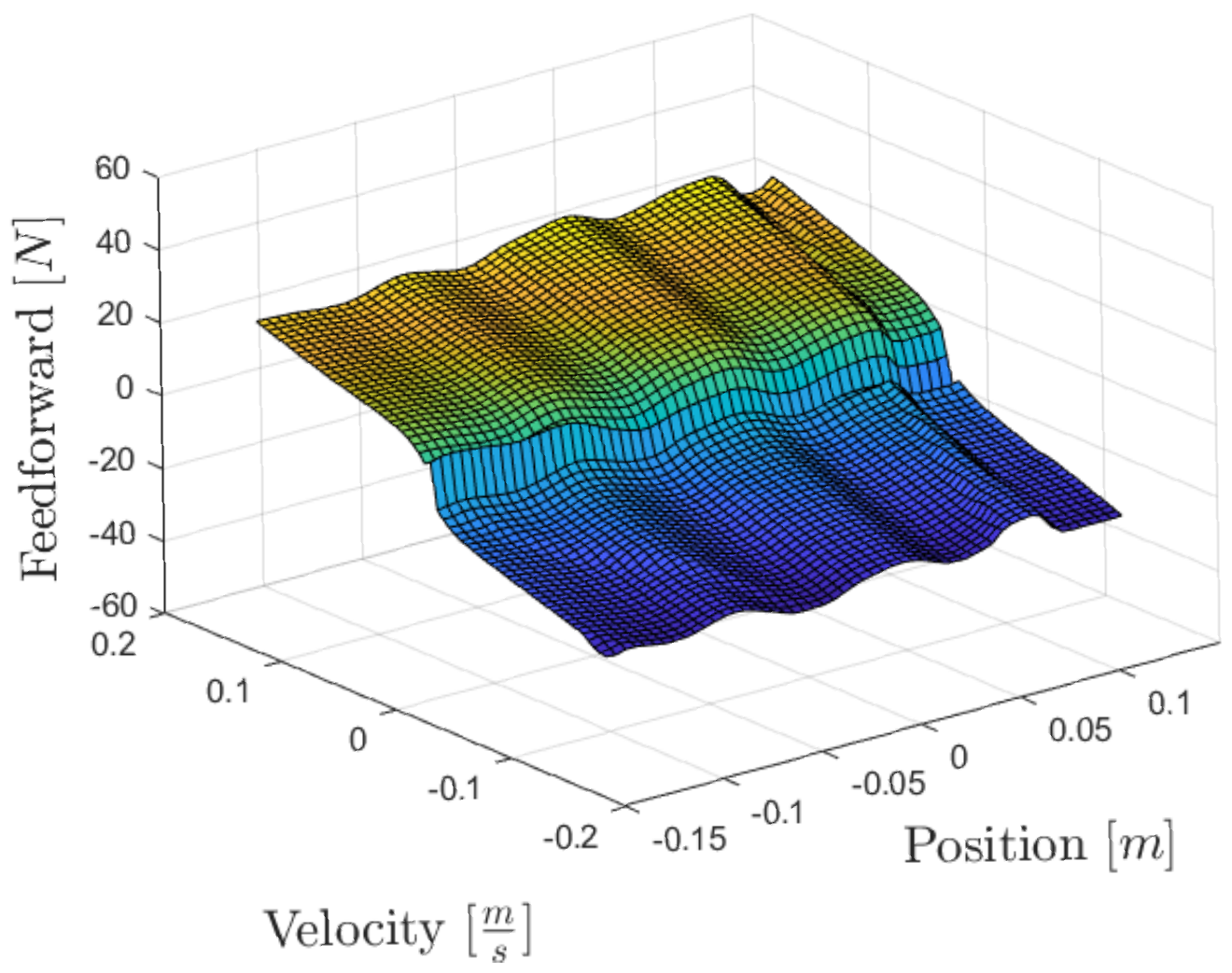}
        \caption{PGNN model with $\gamma = 0.1$.}
         \label{fig:FrictionModels_PGNN_B}
    \end{subfigure}
    \caption{Friction $F_{\textup{fric}}$ of the CLM as in~\eqref{eq:CLM_ContinuousTime} identified by the PGNN~\eqref{eq:PGNNGeneral} according to~\eqref{eq:IdentificationCriterion} with cost function~\eqref{eq:CostFunctionPGNNExtrapolation}, $\Lambda_{\textup{phy}}$ in~\eqref{eq:RuleOfThumbLambda_phy} with $\epsilon =1$, $\Lambda_{\textup{NN}} = 10^{-5} I$ and $\gamma = 0$ (left window) and $\gamma = 0.1$ (right window).}
    \label{fig:FrictionModels_PGNN}
\end{figure}

The effect of the regularization term~\eqref{eq:RegularizationTermPhysics} is demonstrated by training a PGNN according to~\eqref{eq:IdentificationCriterion} with cost function~\eqref{eq:CostFunctionPGNNExtrapolation} on the CLM data set using $\gamma = 0$ and $\gamma = 0.1$, see Sec.~\ref{sec:ValidationRealLife} for the details. 
Following the above guidelines for hyperparameter tuning, we choose $\Lambda_{\textup{phy}}$ as in~\eqref{eq:RuleOfThumbLambda_phy} with $\epsilon = 1$ and $\Lambda_{\textup{NN}} = 10^{-5} I$. 
Fig.~\ref{fig:FrictionModels_PGNN} visualizes the friction models identified by the PGNN. 
Compared to Fig.~\ref{fig:FrictionModels}, which visualizes the friction identified by the physics--based and NN--based model, we observe that the PGNN recovers both the Coulomb friction, as well as the position dependency.
For $\gamma = 0$ in Fig.~\ref{fig:FrictionModels_PGNN_A}, it is observed that the PGNN still suffers when extrapolating, which is diminished for the PGNN trained with $\gamma = 0.1$ in Fig.~\ref{fig:FrictionModels_PGNN_B}.

\section{Input--to--state stability of PGNN feedforward controllers}
\label{sec:Stability}
Stability of a feedforward controller~\eqref{eq:FeedforwardIdentifiedGeneral} is a prerequisite for safe operation of the closed--loop system, i.e., for bounded reference values $r(k)$, the feedforward input $u_{\textup{ff}}(k)$ must remain bounded.
A linear feedforward controller as in~\eqref{eq:LinearFeedforwardController} is stable when the poles of the transfer function $G^{-1}(q)$ are within the unit circle.
However, for the nonlinear PGNN feedforward~\eqref{eq:PGNNGeneral}, stability is determined by the combination of the physical and NN model, which complicates the assessment of stability. 

For the sake of presentation, we consider a PGNN~\eqref{eq:PGNNGeneral} with a linear physical model, i.e., $f_{\textup{phy}} \big( \theta_{\textup{phy}}, \phi(k) \big) = \theta_{\textup{phy}}^T \phi(k)$, and assume that $T\big( \phi(k) \big) = \phi(k)$.
Hence, the PGNN feedforward controller~\eqref{eq:FeedforwardIdentifiedGeneral} is given as
\begin{equation}
    \label{eq:PGNNFeedforwardLinear}
    u_{\textup{ff}}(k) = \hat{\theta}_{\textup{phy}}^T \phi_{\textup{ff}}(k) + f_{\textup{NN}} \big( \hat{\theta}_{\textup{NN}}, \phi_{\textup{ff}}(k) \big). 
\end{equation}
Let $\phi_{\textup{ff}}(k) = [\phi_{r}(k)^T, \phi_{u_{\textup{ff}}}(k)^T]^T$, with $\phi_r(k) := [r(k+n_k+1), ..., r(k+n_k-n_a+1)]^T$ and $\phi_{u_{\textup{ff}}}(k) := [u_{\textup{ff}}(k-1), ..., u_{\textup{ff}}(k-n_b+1)]^T$. 
Similarly, $\hat{\theta}_{\textup{phy}} = [\hat{\theta}_r^T, \hat{\theta}_{u_{\textup{ff}}}^T]^T$, with $\hat{\theta}_r = [I^{(n_a+1)\times(n_a+1)}, 0^{(n_a+1)\times (n_b-1)}] \hat{\theta}_{\textup{phy}}$ and $\hat{\theta}_{u_{\textup{ff}}} = [0^{(n_b-1)\times(n_a+1)}, I^{(n_b-1)\times(n_b-1)}] \hat{\theta}_{\textup{phy}}$. 
We denote $| \cdot |$ as the element--wise absolute value operator, and $\| \cdot \|$ denotes a vector or matrix norm. 
For a matrix $Q$, we denote its maximum and minimum eigenvalue as $\lambda_{\textup{max}} (Q)$ and $\lambda_{\textup{min}}(Q)$, respectively. 
A function $\kappa : \mathbb{R}_{\geq 0} \rightarrow \mathbb{R}_{\geq 0}$ is a $\mathcal{K}$--function if it is continuous, strictly increasing and $\kappa(0) = 0$. It is a $\mathcal{K}_{\infty}$--function if it is a $\mathcal{K}$--function and $\kappa (s) \to \infty$ as $s \to \infty$.

We rewrite the PGNN feedforward~\eqref{eq:PGNNFeedforwardLinear} into the following state--space representation 
\begin{align}
    \begin{split}
        \label{eq:PGNNStateSpace}
        \phi_{u_{\textup{ff}}}(k+1) & = A (\hat{\theta}_{u_{\textup{ff}}} ) \phi_{u_{\textup{ff}}}(k) + B \left( \hat{\theta}_{r}^T \phi_r(k) + f_{\textup{NN}} \left( \hat{\theta}_{\textup{NN}}, \begin{bmatrix} \phi_r(k) \\ \phi_{u_{\textup{ff}}}(k) \end{bmatrix} \right) \right), \\
        u_{\textup{ff}}(k) & = [1, 0^{1 \times (n_b-2)}] \phi_{u_{\textup{ff}}}(k+1),
    \end{split}
\end{align}
where $B = [1, 0^{1 \times (n_b-2)}]^T$, $A(\hat{\theta}_{u_{\textup{ff}}}) = \begin{bmatrix} \hat{\theta}_{u_{\textup{ff}}}^T \\ \begin{matrix} I^{(n_b-2)\times(n_b-2)} & O^{ (n_b-2) \times 1} \end{matrix} \end{bmatrix}$, and $\phi_r(k)$ is the external input. 
As a consequence, we can employ the discrete--time input--to--state stability (ISS) framework defined in~\cite{Jiang2001}. 

If the state--space PGNN feedforward~\eqref{eq:PGNNStateSpace} is ISS with respect to the external input $\phi_r(k+1)$, the following two desirable properties hold:
\begin{enumerate}
    \item $\phi_{u_{\textup{ff}}}(k)$ remains bounded for bounded $\phi_r(k)$ and, consequently, $u_{\textup{ff}}(k)$ remains bounded;
    \item $\phi_{u_{\textup{ff}}}(k) \to 0$ for $\phi_r(k) \to 0$ and, consequently, $u_{\textup{ff}} \to 0$. 
\end{enumerate}

In order to analyze ISS of the PGNN feedforward controller~\eqref{eq:PGNNStateSpace}, we will use an ISS--Lyapunov function, as defined next. 
\begin{definition}
\label{def:ISSLyapunov}
    A function $V: \mathbb{R}^{n} \rightarrow \mathbb{R}_{\geq 0}$ is called an ISS--Lyapunov function for a discrete--time state--space system $\phi_{u_{\textup{ff}}}(k+1) = f \big( \phi_{u_{\textup{ff}}}(k), \phi_r(k+1) \big)$ if the following conditions hold:
    \begin{enumerate}
        \item There exists $\mathcal{K}_{\infty}$--functions $\kappa_1$, $\kappa_2$ such that 
        \begin{equation}
            \label{eq:ISSLyapunov1}
            \kappa_1 \big(\|\phi_{u_{\textup{ff}}}(k)\| \big) \leq V \big( \phi_{u_{\textup{ff}}}(k) \big) \leq \kappa_2 \big( \|\phi_{u_{\textup{ff}}}(k)\| \big), \; \forall \, \phi_{u_{\textup{ff}}}(k) \in \mathbb{R}^{n_b-1}. 
        \end{equation}
        \item There exists a $\mathcal{K}_{\infty}$--function $\kappa_3$ and a $\mathcal{K}$--function $\sigma$, such that
        \begin{align}
        \begin{split}
            \label{eq:ISSLyapunov2}
            V & \big( \phi_{u_{\textup{ff}}}(k+1) \big) - V \big( \phi_{u_{\textup{ff}}}(k) \big) \leq - \kappa_3 (\|\phi_{u_{\textup{ff}}}(k) \| ) + \sigma(\|\phi_r(k+1) \| ) , \\
            & \quad \forall \; \phi_{u_{\textup{ff}}}(k) \in \mathbb{R}^{n_b-1}, \; \forall \; \phi_r(k) \in \mathbb{R}^{n_a+1}. 
        \end{split}
        \end{align}
    \end{enumerate}
\end{definition}

\begin{remark}
\label{re:LipschitzNN}
    For all commonly applied activation functions, the NN~\eqref{eq:NNParametrization} has bounded partial derivatives, i.e., there exists an $K \in \mathbb{R}^{n_a+n_b}$ such that
    \begin{equation}
        \label{eq:NNLipschitz}
        \left| f_{\textup{NN}} \big( \hat{\theta}_{\textup{NN}}, \phi_{\textup{ff}}^A (k) \big) - f_{\textup{NN}} \big( \hat{\theta}_{\textup{NN}}, \phi_{\textup{ff}}^B(k) \big) \right| \leq K^T \left| \phi_{\textup{ff}}^A(k) - \phi_{\textup{ff}}^B (k) \right|. 
    \end{equation}
    Using backpropagation, it is possible to find values for $K$, e.g.,
    \begin{align}
    \begin{split}
        \label{eq:NNLipschitzValue}
        K^T & = \max_{\phi_{\textup{ff}}(k)} \frac{\partial f_{\textup{NN}} \big( \hat{\theta}_{\textup{NN}}, \phi_{\textup{ff}}(k) \big) }{\partial \phi_{\textup{ff}}(k) } \\
        & = \max_{\phi_{\textup{ff}}(k)} \hat{W}_{L+1} \textup{diag} (\alpha_l') \hdots \hat{W}_2 \textup{diag} ( \alpha_1') \hat{W}_1 \\
        & \leq  | W_{L+1} | \; \Pi_{l=1}^{L} | \hat{W}_l | \;  \max_{\phi_{\textup{ff}}(k)} \big( \textup{diag} (\alpha_i') \big) = \Pi_{l=1}^{L+1} | W_l|,
    \end{split}
    \end{align}
    where $\alpha_l' = \frac{\partial \alpha_l(x)}{\partial x} \big|_{x = x_{l}}$, with $x_l$ the input to NN layer $l$.
    The final term is obtained by substition of $\max_{\phi_{\textup{ff}}(k)} \big( \textup{diag} (\alpha_i') \big) = I$, which holds for $\tanh$, ReLU, and several other activation functions.
    In addition, we let $K_{r} := [I^{(n_a+1)\times(n_a+1)}, 0^{(n_a+1) \times (n_b-1)}] K$, and $K_{u_{\textup{ff}}} := [0^{(n_b-1)\times(n_a+1)}, I^{(n_b-1)\times(n_b-1)}]K$, such that $K = [K_r^T, K_{u_{\textup{ff}}}^T]^T$.  
\end{remark}

\begin{remark}
    A global ISS result for black--box NNs was derived in~\cite{Bonassi2021} by using a scalar Lipschitz constant.
    We reduce conservatism by retaining $K \in \mathbb{R}^{n_a+n_b}$, since, as it turns out, ISS is determined by $K_{u_{\textup{ff}}}$. This results in less conservative Lipschitz bounds, which apply to black--box NNs as well.
\end{remark}

\begin{assumption}
\label{as:EquilibriumOrigin}
    The origin $\phi_{u_{\textup{ff}}}(k) = 0$ is an equilibrium for the unexcited PGNN~\eqref{eq:PGNNStateSpace}, i.e., $\phi_r(k) = 0$, such that~\eqref{eq:PGNNStateSpace} gives
    \begin{equation}
        \label{eq:EquilibriumOrigin}
        f_{\textup{NN}} \big(\hat{\theta}_{\textup{NN}}, 0 \big) = 0. 
    \end{equation}
\end{assumption}
Note that, we can introduce a coordinate transformation, e.g., $\zeta(k) = \phi_{u_{\textup{ff}}}(k) + \varepsilon$ to satisfy Assumption~\ref{as:EquilibriumOrigin}. 
\begin{assumption}
\label{as:StablePhysicalModel}
    There exists $P \succ 0$ such that \[Q := P - A(\hat{\theta}_{u_{\textup{ff}}})^T P A (\hat{\theta}_{u_{\textup{ff}}}) \succ 0.\] 
\end{assumption}
Assumption~\ref{as:StablePhysicalModel} requires that the parameters $\hat{\theta}_{u_{\textup{ff}}}$ are such that the physical model is stable, i.e., $A(\hat{\theta}_{u_{\textup{ff}}})$ is a Schur matrix. 
A $(P,Q)$ pair satisfying Assumption~\ref{as:StablePhysicalModel} is typically obtained by choosing $Q \succ 0$ and solving the discrete--time Lyapunov equation to obtain $P \succ 0$.

\begin{theorem}[PGNN feedforward ISS]
\label{th:StabilityPGNN}
    Consider the PGNN feedforward controller~\eqref{eq:PGNNFeedforwardLinear}, and its state--space representation~\eqref{eq:PGNNStateSpace}.
    Suppose that Assumptions~\ref{as:EquilibriumOrigin} and~\ref{as:StablePhysicalModel} hold. Let $(P,Q)$ satisfy Assumption~\ref{as:StablePhysicalModel}, and define
    \begin{equation}
    \label{eq:StabilityPGNNConditionVariable}
        c_{\beta} := B^T P \left( I + \frac{1}{\beta \lambda_{\textup{min}}(Q)} A(\hat{\theta}_{u_{\textup{ff}}}) A(\hat{\theta}_{u_{\textup{ff}}})^T P  \right) B. 
    \end{equation}
    Then, if there exists a $\beta > 0$ such that
    \begin{equation}
        \label{eq:StabilityPGNNCondition}
        K_{u_{\textup{ff}}}^T K_{u_{\textup{ff}}} < \frac{(1-\beta) \lambda_{\textup{min}} (Q)}{c_{\beta} },
    \end{equation}
    the PGNN feedforward state--space representation~\eqref{eq:PGNNStateSpace} is ISS.
\end{theorem}
\begin{proof}
    See~\ref{app:ThPGNNStability}. 
\end{proof}

The value of $\beta$ for which the right--hand side in~\eqref{eq:StabilityPGNNCondition} is maximal, is 
\begin{equation}
    \label{eq:OptimalChoiceForAlpha}
    \beta = \frac{- B^T P A A^T P B + \sqrt{B^T P \big( \lambda_{\textup{min}}(Q) I + A A^T P \big) B B^T P A A^T P B}}{\lambda_{\textup{min}}(Q) B^T P B}. 
\end{equation}
Note that, since $B$ is a column, the square root and division are scalar operations. 
Eq.~\eqref{eq:OptimalChoiceForAlpha} is obtained by setting the derivative w.r.t. $\beta$ of the right hand side of~\eqref{eq:StabilityPGNNCondition} equal to zero, and choosing the option for which $\beta > 0$. 

\begin{remark}
    A similar result as in Theorem~\ref{th:StabilityPGNN} can be obtained for a PGNN with a general nonlinear physical model when a quadratic Lyapunov function is available for the physical model. 
    Similarly, the transform $T(\cdot )$ can be included by including it for computation of $K$ in~\eqref{eq:NNLipschitzValue}. 
\end{remark}

The ISS condition~\eqref{eq:StabilityPGNNCondition} in Theorem~\ref{th:StabilityPGNN} can be validated \emph{after} training using $\beta$ as in~\eqref{eq:OptimalChoiceForAlpha}, the upperbound of $K_{u_{\textup{ff}}} = [0^{(n_b-1)\times (n_a+1)}, I^{(n_b-1)\times (n_b-1)}] K$ with $K$ in~\eqref{eq:NNLipschitzValue} for some pair $(P, Q)$. 
With the aim to ensure \emph{before} training that the PGNN is ISS, we fix $\hat{\theta}_{r} = \theta_r^*$, and constraint the network weights to satisfy the ISS condition~\eqref{eq:StabilityPGNNCondition}. 
\begin{lemma}[Training imposed ISS]
\label{le:StabilityAPrioriWithAssumption}
    Consider the PGNN feedforward controller with linear physical model, such that it admits a state--space representation of the form~\eqref{eq:PGNNStateSpace}.
    Suppose that Assumption~\ref{as:EquilibriumOrigin} holds, that a $(P,Q)$ pair satisfying Assumption~\ref{as:StablePhysicalModel} is available, and choose $\beta$ as in~\eqref{eq:OptimalChoiceForAlpha}.
    Define the set 
    \begin{align} 
    \begin{split}
        \label{eq:SetOptimization}
        \Theta := \Bigg\{& \theta \in \mathbb{R}^{n_{\theta}} \; \Bigg| \;  \big( \theta_r = \theta_r^* \big) \; \wedge  \\
        & \left( \left\| \, \big( \Pi_{l=1}^{L+1} | W_l |\big) \begin{bmatrix}  0^{(n_a+1)\times(n_b-1)} \\  I^{(n_b-1)\times(n_b-1)} \end{bmatrix} \, \right\|_2^2 < \frac{(1-\beta) \lambda_{\textup{min}}(Q)}{c_{\beta}}\right) \Bigg\}, 
    \end{split}
    \end{align}
    and train the PGNN according to identification criterion
    \begin{align}
    \begin{split}
        \label{eq:IdentificationCriterionStable}
        \hat{\theta} = \textup{arg} \min_{\theta \; \in \; \Theta} V \big( \theta, Z^N \big).
    \end{split}
    \end{align}
    Then, the training returns an ISS PGNN feedforward controller~\eqref{eq:FeedforwardIdentifiedGeneral},~\eqref{eq:PGNNGeneral}.
\end{lemma}
\begin{proof}
    See~\ref{app:leStability1}. 
\end{proof}

In the remainder of this section, we consider the situation in which Assumption~\ref{as:StablePhysicalModel} is violated, i.e., $A(\hat{\theta}_{u_{\textup{ff}}})$ is not Schur.
To obtain a stable PGNN feedforward controller, we take inspiration from the stable inversion techniques for linear systems as in~\cite{Zundert2018}, which apply either a non--causal feedforward controller design, or a stable approximate inversion such as ZPETC, ZMETC, or NPZ--ignore. 

We define $n_{\textup{us}} \in \mathbb{Z}_{\geq 0}$ as the number of unstable eigenvalues of $A(\theta_{u_{\textup{ff}}}^*)$. 
Subsequently, we adjust the PGNN~\eqref{eq:PGNNFeedforwardLinear} to include a preview window of $n_{\textup{pw}} \in \mathbb{Z}_{\geq 0}$ while reducing the number of past inputs by $n_{\textup{us}}$, such that
\begin{align}
\begin{split}
    \label{eq:PGNNExtendedPreviewWindow}
    \hat{u} \big( \theta, \tilde{\phi}(k) \big) & = [\theta_{r}^T, \theta_{u_{\textup{ff}}}^T] \begin{bmatrix} \tilde{\phi}_y(k) \\ \tilde{\phi}_u(k) \end{bmatrix} + f_{\textup{NN}} \left(\theta_{\textup{NN}}, \begin{bmatrix} \tilde{\phi}_y(k) \\ \tilde{\phi}_u(k) \end{bmatrix} \right), \\
    \tilde{\phi}_y(k) & := [y(k+n_k+n_{\textup{pw}} +1), ..., y(k+n_k-n_a+1)]^T, \\
    \tilde{\phi}_u(k) & := [u(k-1), ..., u(k-n_b+n_{\textup{us}}+1)]^T.
\end{split}
\end{align}
After identification of the parameters $\hat{\theta}$ of the PGNN with extended preview~\eqref{eq:PGNNExtendedPreviewWindow} and using it for feedforward control, it can be rewritten into state--space representation, similar as was done in~\eqref{eq:PGNNStateSpace} for the original PGNN~\eqref{eq:PGNNFeedforwardLinear}, such that we obtain
\begin{align}
\begin{split}
    \label{eq:PGNNExtendedPreviewWindowStateSpace}
    \tilde{\phi}_{u_{\textup{ff}}} (k+1) & = \tilde{A}(\hat{\theta}_{u_{\textup{ff}}}) \tilde{\phi}_{u_{\textup{ff}}} (k) + \tilde{B} \left( \hat{\theta}_r^T \tilde{\phi}_r(k) +  f_{\textup{NN}} \left( \hat{\theta}_{\textup{NN}}, \begin{bmatrix} \tilde{\phi}_r(k) \\ \tilde{\phi}_{u_{\textup{ff}}}(k) \end{bmatrix} \right) \right), \\
    u_{\textup{ff}}(k) & = [1, 0^{1 \times (n_b-2-n_{\textup{us}})}] \tilde{\phi}_{u_{\textup{ff}}}(k+1). 
\end{split}
\end{align}

For the PGNN feedforward controller with extended preview window~\eqref{eq:PGNNExtendedPreviewWindowStateSpace}, we impose ISS during the PGNN identification following any of the next approaches:
\begin{enumerate}
    \item \emph{Complete retraining:} re--identify the linear part $\theta_{\textup{phy}}^* = [{\theta_{r}^*}^T, {\theta_{u_{\textup{ff}}}^*}^T]$ according to~\eqref{eq:CostFunctionMSE} with $n_{\textup{pw}}$ sufficiently large to have $\tilde{A}(\theta_{u_{\textup{ff}}}^*)$ satisfy Assumption~\ref{as:StablePhysicalModel}, and re--identify $\hat{\theta}$ according to Lemma~\ref{le:StabilityAPrioriWithAssumption};
    \item \emph{Partial retraining:} fix $\theta_{u_{\textup{ff}}}^*$ to have $\tilde{A}(\theta_{u_{\textup{ff}}}^*)$ retain the stable eigenvalues of $A(\theta_{u_{\textup{ff}}}^*)$ corresponding to the original PGNN~\eqref{eq:PGNNFeedforwardLinear}, and re--identify $\hat{\theta}$ according to Lemma~\ref{le:StabilityAPrioriWithAssumption} without any conditions on $n_{\textup{pw}}$;
    \item \emph{Stable approximation:} train $\hat{\theta}$ for the original PGNN~\eqref{eq:PGNNFeedforwardLinear} while ensuring that, after application of a stable approximation method to the linear part, the PGNN~\eqref{eq:PGNNExtendedPreviewWindow} is ISS following Lemma~\ref{le:StabilityAPrioriWithAssumption} with $\Theta$ in~\eqref{eq:SetOptimization} adjusted accordingly.
\end{enumerate}
The first approach has the most flexibility during training, but requires a sufficiently large $n_{\textup{pw}}$, which typically depends on the location of the unstable eigenvalues of $A(\hat{\theta}_{u_{\textup{ff}}})$ of the original PGNN~\eqref{eq:PGNNFeedforwardLinear}. 
On the other hand, the second and third approach guarantee stability for any $n_{\textup{pw}}$, where the second relies on the data to find a stable inverse, while the third relies on the stable approximation method.

\section{Experimental and simulation validation}
\label{sec:Validation}
The (PG)NNs in this section are trained using MATLAB's ``lsqnonlin" optimization tool with Levenberg--Marquardt and dedicated functions for the (PG)NNs.
All (PG)NNs have a single hidden layer with $\tanh$--activation function, and the number of neurons $n_1$ are found by performing an unregularized training, i.e., minimizing~\eqref{eq:CostFunctionMSE}, for increasing number of neurons $n_1$ until the cost function no longer decreased significantly.
The data sets are randomly split into $70$$\%$ training data which is used for training, and $30$$\%$ validation data which is used only for early--stopping during training when overfitting is detected. 
We train each (PG)NN $10$ times with random (with input normalization) initialization of $\theta_{\textup{NL}}^{(0)}$ and optimal initialization~\eqref{eq:OptimizedInitialization} for $\theta_{\textup{L}}^{(0)}$ and select the training that reached the smallest value of the cost function evaluated over the validation data.

\begin{figure}
    \centering
    \includegraphics[width=1.0\linewidth]{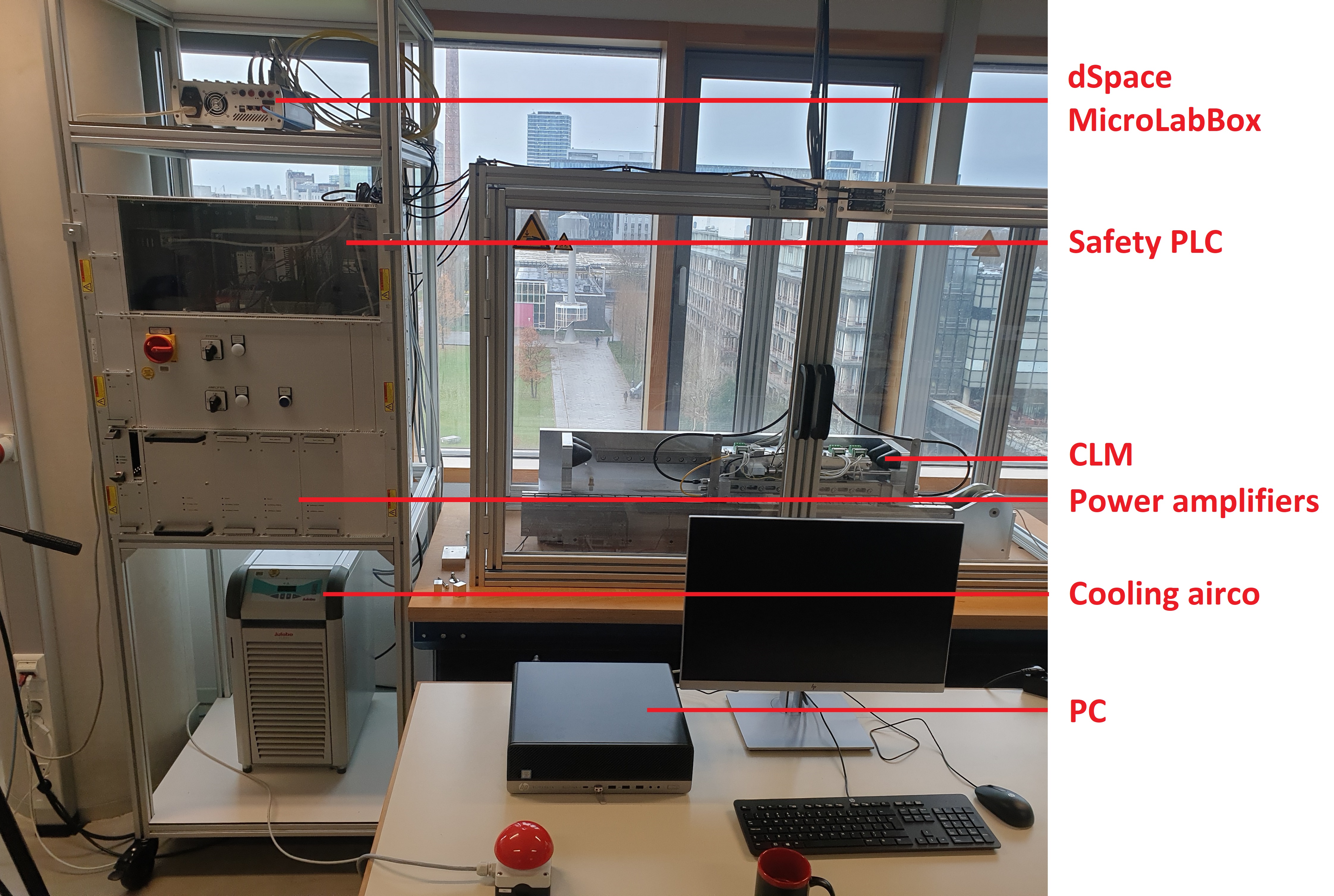}
    \includegraphics[width=1.0\linewidth]{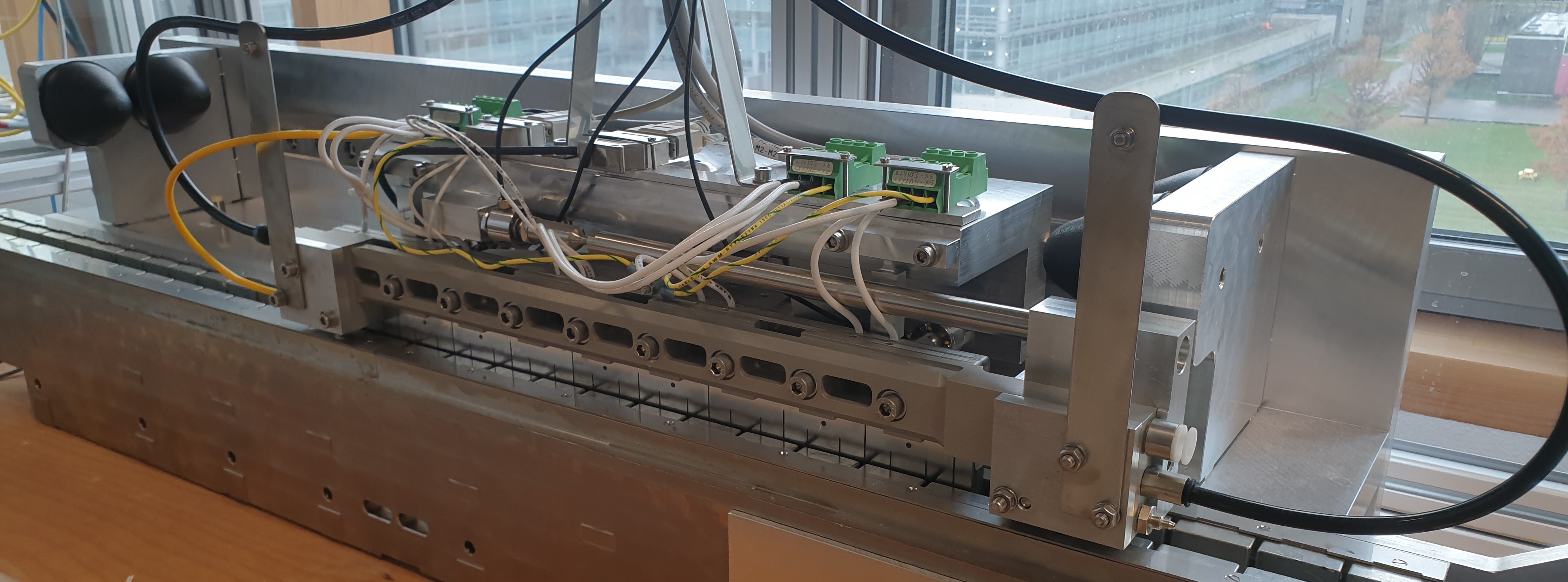}
    \caption{Complete industrial coreless linear motor setup (top window), with an enlarged view of the coreless linear motor (bottom window).}
    \label{fig:CLM}
\end{figure}

\subsection{Real--life industrial coreless linear motor}
\label{sec:ValidationRealLife}
\emph{\textbf{Experimental setup:}} we consider the problem of position control of the industrial coreless linear motor (CLM) displayed in Fig.~\ref{fig:CLM}, which was formerly part of the longstroke actuation in a lithography machine. For research purposes, the CLM is limited to exert forces up to $500$ $N$. 
The CLM contains three coil sets, each consisting of three coils connected in star configuration that are powered by a three--phase power amplifier. 
The system is controlled by a dSPACE MicroLabBox that receives encoder position measurements with an accuracy of $5 \cdot 10^{-6}$ $m$, computes the control input, and converts them into current setpoints via a commutation algorithm to be send to the power amplifiers. 
The PC is used to program software in MATLAB/Simulink that is uploaded to the dSPACE MicroLabBox. 
Relevant signals are accessed during run--time using the dSPACE ControlDesk software. 
The coils are cooled using water and the cooling unit.
Finally, the safety programmable logic controller (PLC) cuts of the power to the amplifiers when safety is at risk, e.g., when the stop button is pushed, or when the CLM access doors are opened.

\emph{\textbf{System modelling:}} 
a physical model of the CLM displayed in Fig.~\ref{fig:CLM} was derived based on Newton's second law in~\eqref{eq:CLM_ContinuousTime}.
We discretize the system by approximating $\dot{y}(t) \approx \delta y(k)$, with $\delta = \frac{q-q^{-1}}{2T_s}$, $\ddot{y}(t) \approx \delta^2 y(k)$ and approximate the ZOH D2C delay by $\Delta = \frac{q+1}{2}$, such that~\eqref{eq:CLM_ContinuousTime} becomes
\begin{equation}
    \label{eq:CLM_DiscreteTime}
    u(k) = \Delta \left( m \delta^2 y(k) + F_{\textup{fric}} \big( y(k), \delta y(k) \big) \right). 
\end{equation}
By substituting $\delta$ and $\Delta$ and using $q y(k) = y(k+1)$, we observe from~\eqref{eq:CLM_DiscreteTime} that $n_a = 5$, $n_b = 1$ and $n_k = 2$ in $\phi(k)$. 
The CLM is operated in closed--loop by the ZOH discretization of the feedback controller
\begin{equation}
    \label{eq:CLMFeedbackController}
    C(s) = \frac{1.056\cdot 10^{8} s^2 + 2.282\cdot 10^{9} s + 7.884\cdot 10^9}{s^3 + 547.4s^2 + 7.643\cdot 10^4 s - 0.0001669},
\end{equation}
which was manually loopshaped using frequency domain data.

\begin{figure}
    \centering
    \includegraphics[width=0.8\linewidth]{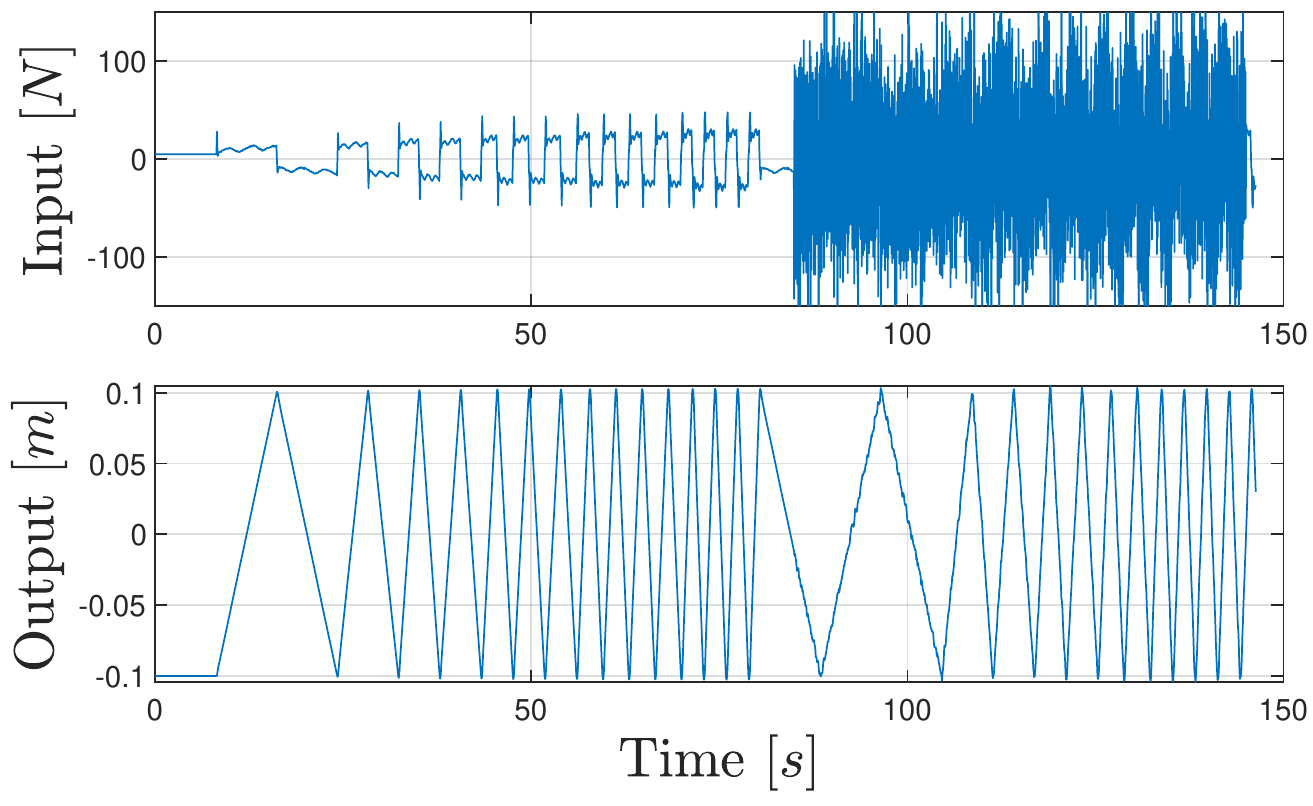}
    \caption{Visualization of the data set $Z^N$ generated on the CLM with input $u(k)$ (top window) and output $y(k)$ (bottom window).}
    \label{fig:DataSet}
\end{figure}
\emph{\textbf{Training data generation:}}
The data set $Z^N$ is generated in closed--loop by sampling $u(k)$ and $y(k)$ at a frequency of $1$ $kHz$, while exciting the system via:
\begin{enumerate}
    \item \emph{The reference} $r(k)$ using a third order reference moving from $-0.1$ to $0.1$~$m$ with acceleration $1$ $\frac{m}{s^2}$, jerk $1000$ $\frac{m}{s^3}$ and different velocities $n 0.025$ $\frac{m}{s}$, $n \in \{ 1, ..., 6\}$;
    \item \emph{The input} $u(k)$ using a zero--mean white noise with variance $50$ $N^2$ added to the input during the final half of the experiment.
\end{enumerate} 
There is not yet any feedforward controller available and used. The resulting input $u(k)$ and output $y(k)$ are visualized in Fig.~\ref{fig:DataSet}. 
The experiment ran for $146$~$s$ yielding $N = 146 \cdot 10^3$ data samples.

\emph{\textbf{Feedforward controllers:}}
We consider the following model parametrizations for the design of the feedforward controller. 
Firstly, the physics--based model derived from first--principle knowledge
\begin{equation}
    \label{eq:CLM_PhysicsBased}
    \hat{u} \big( \theta_{\textup{phy}}, \phi(k) \big) = m \Delta \delta^2y(k) + f_v \Delta \delta y(k) + f_c \Delta \textup{sign} \big( \delta y(k) \big),
\end{equation}
where $\theta_{\textup{phy}} = [m, f_v, f_c]^T$, with $m \in \mathbb{R}_{>0}$ the mass, $f_v \in \mathbb{R}_{>0}$ the viscous friction coefficient, and $f_c \in \mathbb{R}_{>0}$ the Coulomb friction coefficient.
Secondly, we consider a NN--based model, which is given as
\begin{equation}
    \label{eq:CLM_NNBased}
    \hat{u} \big( \theta_{\textup{NN}}, \phi(k) \big) = W_2 \tanh \big( W_1 \Delta [y(k+2), ..., y(k-2)]^T + B_1 \big) +B_2,
\end{equation}
with $n_1 = 24$ neurons.
Finally, the PGNN--based model is given as
\begin{align}
\begin{split}
    \label{eq:CLM_PGNNBased}
    \hat{u} \big( \theta, \phi(k) \big) = & m \Delta \delta^2 y(k) + f_v \Delta \delta y(k) + f_c \Delta \textup{sign} \big( \delta y(k) \big) \\
    & + W_2 \tanh \big( W_1 \Delta [y(k), \delta y(k), \delta^2 y(k)]^T + B_1 \big) + B_2. 
\end{split}
\end{align}
which also has $n_1 = 24$ neurons. 
As a result, the physics--based model~\eqref{eq:CLM_PhysicsBased} has $n_{\theta_{\textup{phy}}} = 3$, the NN--based model~\eqref{eq:CLM_NNBased} has $n_{\theta_{\textup{NN}}} = 5 \cdot 24 + 24 + 24+1 = 169$, and the PGNN--based model~\eqref{eq:CLM_PGNNBased} has $n_{\theta} = 3 + 3\cdot 24 + 24 + 24 + 1= 124$ parameters.
Based on these parametrizations, we train and apply the following state--of--the--art feedforward controllers:
\begin{enumerate}
    \item \emph{Physics--based feedforward} using the physics--based model~\eqref{eq:CLM_PhysicsBased} identified according to~\eqref{eq:IdentificationCriterion} with cost function~\eqref{eq:CostFunctionMSE};
    \item \emph{NN--based feedforward} using the NN--based model~\eqref{eq:CLM_NNBased} identified according to~\eqref{eq:IdentificationCriterion} with~\eqref{eq:CostFunctionPGNN} using $\Lambda_{\textup{NN}} = 3.2\cdot 10^{-12} I$;
    \item \emph{PINN--based feedforward} using the NN--based model~\eqref{eq:CLM_NNBased} identified with~\eqref{eq:CostFunctionPINN} with $c=0.5$, and, in addition, the regularization terms in~\eqref{eq:CostFunctionPGNNExtrapolation} with $\Lambda_{\textup{NN}} = 3.2\cdot 10^{-12} I$ and $\gamma = 0.1$;
    \item \emph{PGNN--based feedforward} using the PGNN--based model~\eqref{eq:CLM_PGNNBased} identified with~\eqref{eq:CostFunctionPGNN} using $\Lambda_{\textup{phy}}$ as in~\eqref{eq:RuleOfThumbLambda_phy} with $\epsilon = 1$ and $\Lambda_{\textup{NN}} = 10^{-5} I$;
    \item \emph{PGNN--based feedforward} using the PGNN--based model~\eqref{eq:CLM_PGNNBased} identified with~\eqref{eq:CostFunctionPGNNExtrapolation} using $\Lambda_{\textup{phy}}$ as in~\eqref{eq:RuleOfThumbLambda_phy} with $\epsilon = 1$ and $\Lambda_{\textup{NN}} = 10^{-5} I$ and $\gamma = 0.1$.  
\end{enumerate}

\begin{figure}
    \centering
    \includegraphics[width=1.0\linewidth]{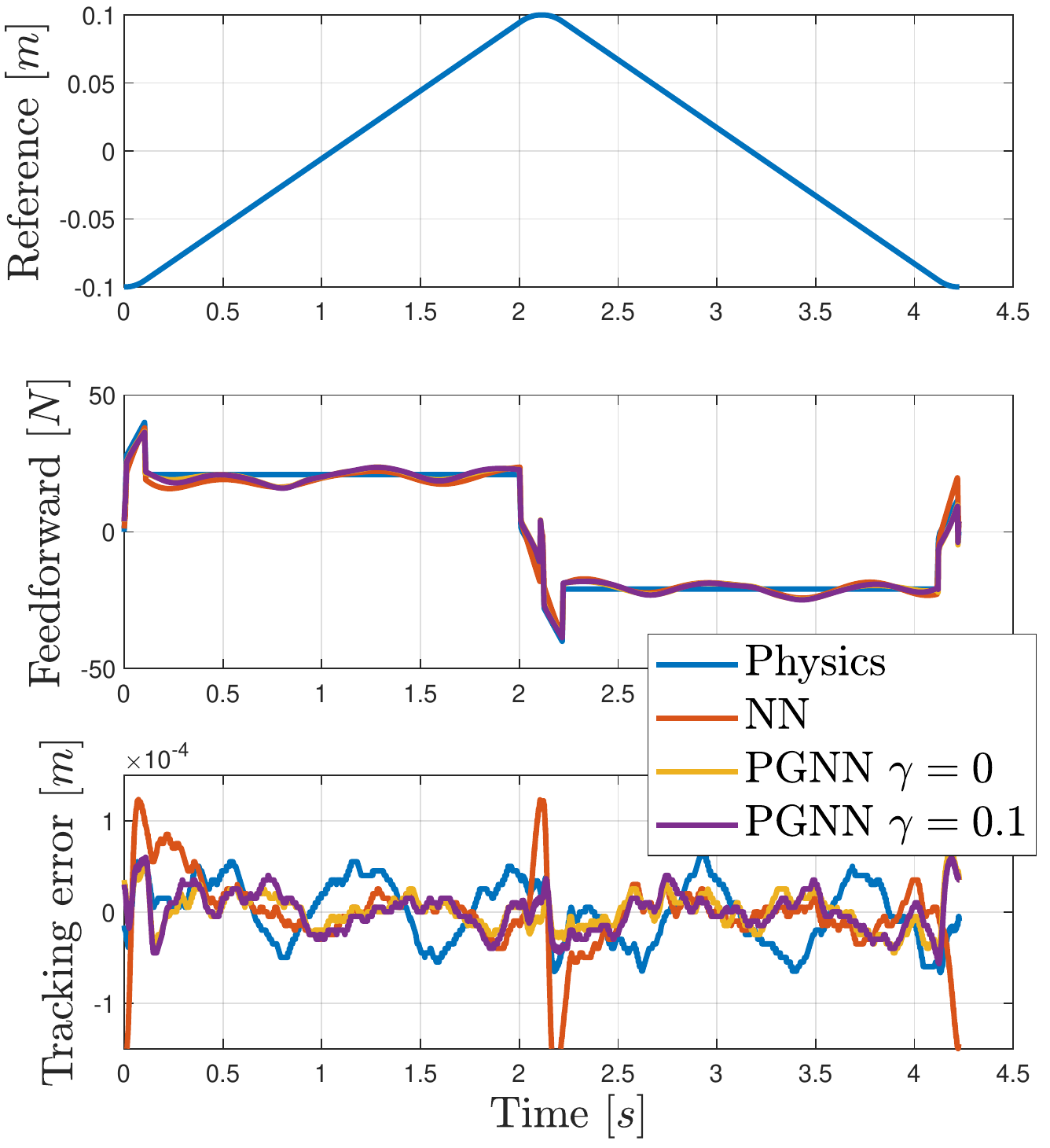}
    \caption{Reference (top window) with the generated feedforward signals (middle window) and the resulting tracking errors (bottom window).}
    \label{fig:CLM_TrackingError1}
\end{figure}
\begin{figure}
    \centering
    \includegraphics[width=1.0\linewidth]{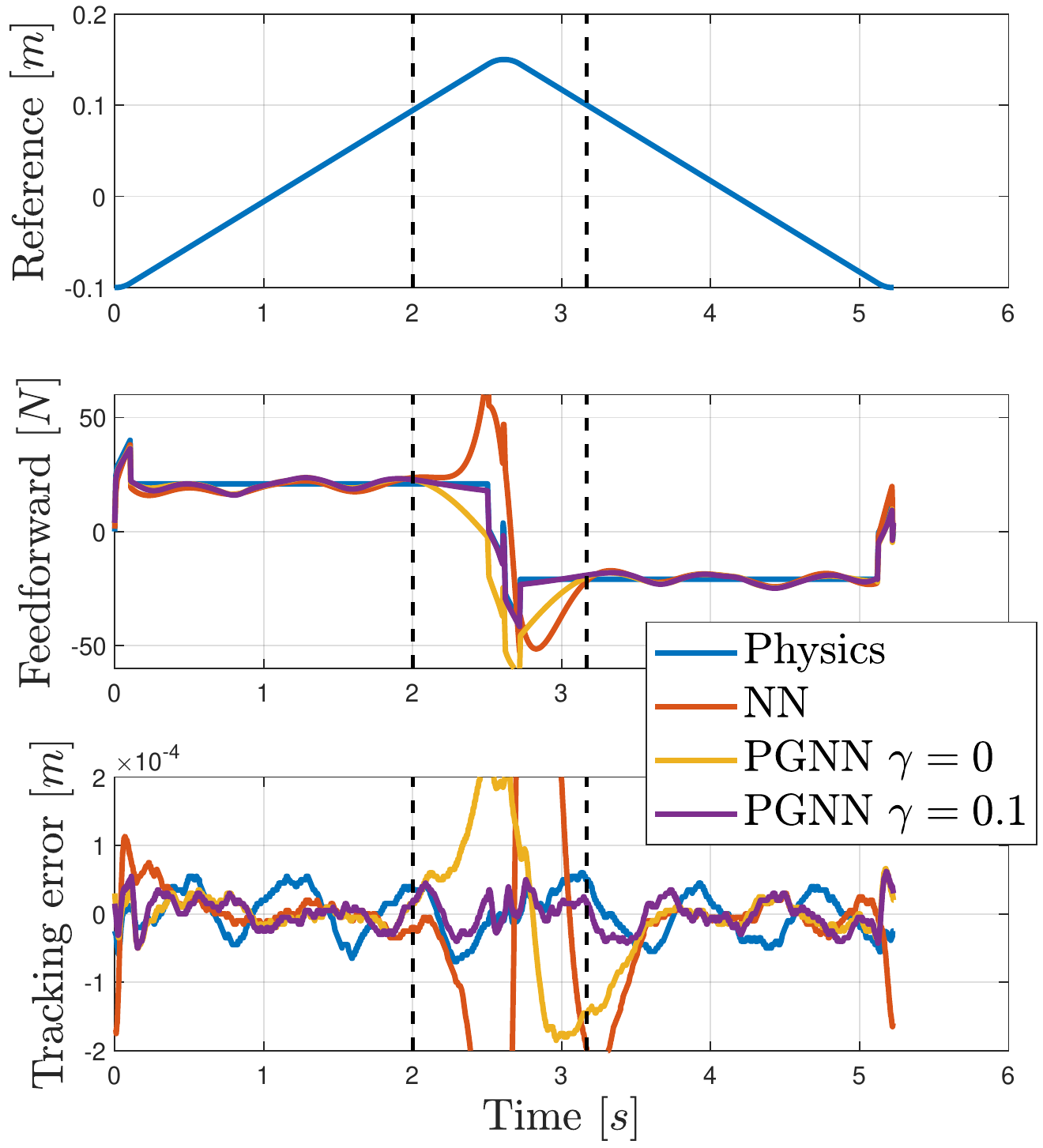}
    \caption{Reference (top window) with the generated feedforward signals (middle window) and the resulting tracking errors (bottom window).}
    \label{fig:CLM_TrackingError2}
\end{figure}

\emph{\textbf{Results:}} Fig.~\ref{fig:CLM_TrackingError1} shows the feedforward signals generated by the physics--based, NN--based and PGNN--based (with $\gamma = 0$ and $\gamma = 0.1$) feedforward controllers and the resulting tracking errors for a reference moving from $-0.1$ to $0.1$ with velocity $0.1$ $\frac{m}{s}$, acceleration $1$ $\frac{m}{s^2}$ and jerk $1000$ $\frac{m}{s^3}$. 
The physics--based feedforward controller experiences a significant tracking error during the constant velocity part of the reference, which is explained by its inability to identify the position dependency of the CLM dynamics as was shown in Fig.~\ref{fig:FrictionModels}.
On the other hand, the NN--based feedforward controller achieves a small tracking error during the constant velocity part, but suffers a loss of performance when accelerating and changing direction. Note that, in Fig.~\ref{fig:FrictionModels} we observed that the NN failed to identify the Coulomb friction term.
The PGNNs seem to combine the best of both approaches: a small tracking error during acceleration as does the physics--based feedforward, and a small tracking error during constant velocity as does the NN--based feedforward controller.

In order to illustrate the effect of the regularization term~\eqref{eq:RegularizationTermPhysics} in enhancing robustness of the feedforward controller, we apply the feedforward to a reference that exceeds the training data set in terms of position, i.e., $\max \big( r(k) \big) = 0.15$ $m$, compared to $\max (y_i) = 1.019$ $m$ in the data set $Z^N$, see also Fig.~\ref{fig:ExtrapolationDataSet}. 
For this reference, Fig.~\ref{fig:CLM_TrackingError2} shows the generated feedforward signals and the resulting tracking errors.
The physics--based feedforward controller does not suffer in performance when operated on this reference.
In contrary, the NN and PGNN with $\gamma = 0$ do not extrapolate to positions outside the training data set, as is observed by the large deviation of the feedforward signals and the resulting increase in the tracking error.
Regularization of the PGNN with $\gamma = 0.1$ in the cost function~\eqref{eq:CostFunctionPGNNExtrapolation} mitigates this issue by promoting compliance with the physical model.

\begin{figure}
    \centering
    \begin{subfigure}{0.65\linewidth}
        \includegraphics[width=1\linewidth]{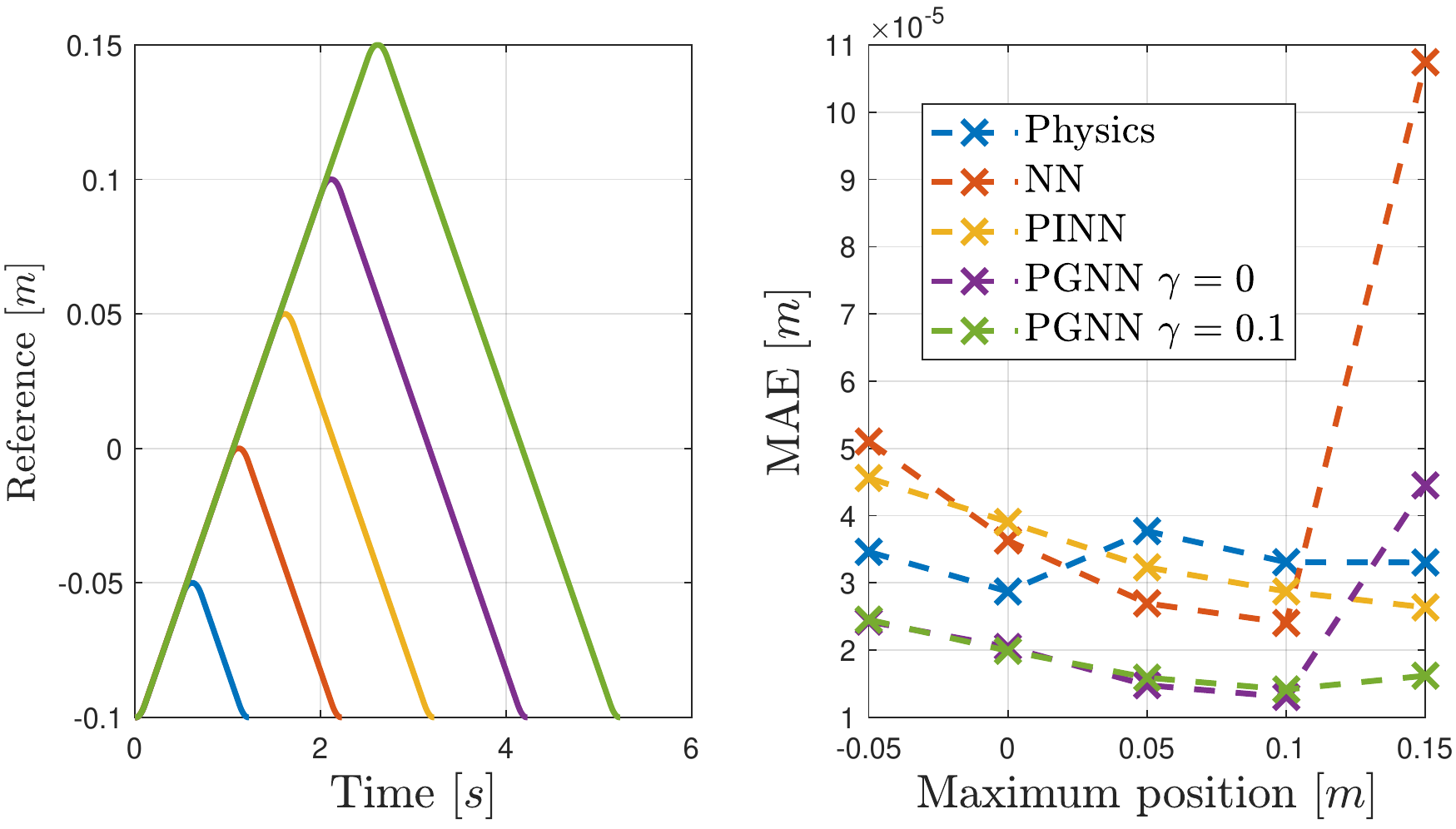}
        \caption{Varying reference position.}
        \label{fig:Robustness_VaryingPosition}
    \end{subfigure}
    \begin{subfigure}{0.65\linewidth}
        \includegraphics[width=1\linewidth]{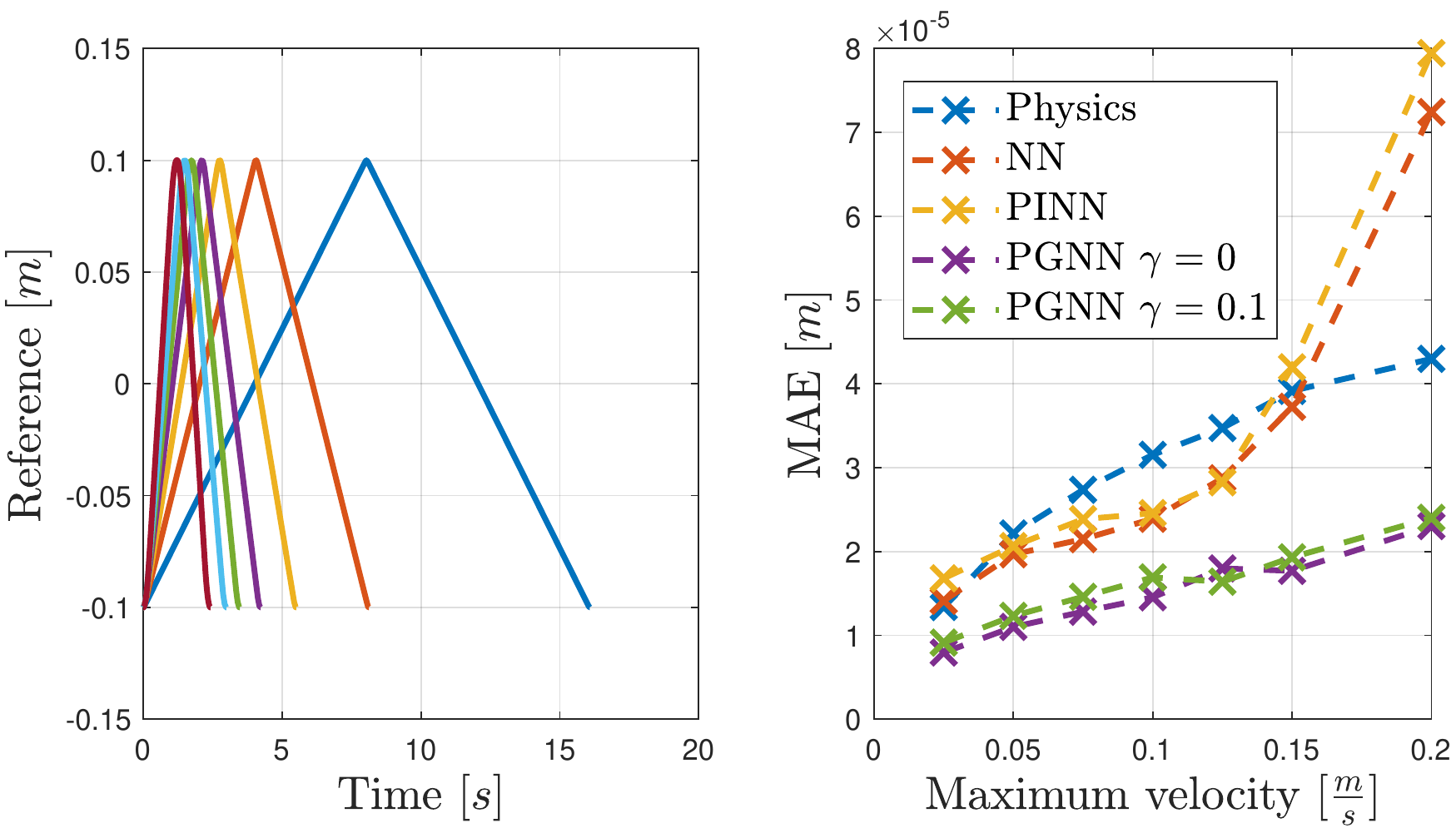}
        \caption{Varying reference velocity.}
        \label{fig:Robustness_VaryingVelocity}
    \end{subfigure}
    \begin{subfigure}{0.65\linewidth}
        \includegraphics[width=1\linewidth]{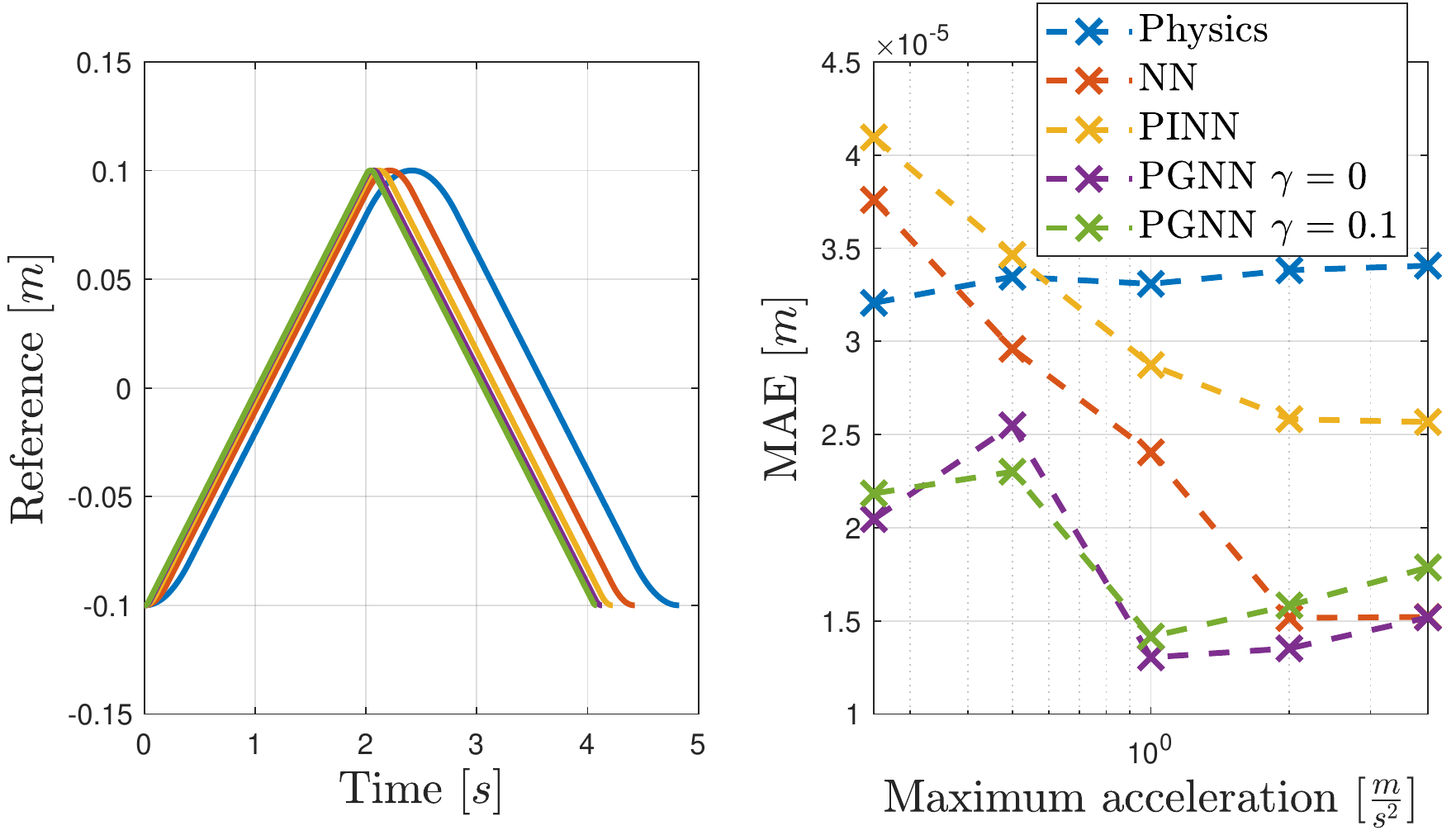}
        \caption{Varying reference acceleration.}
        \label{fig:Robustness_VaryingAcceleration}
    \end{subfigure}
    \caption{References (left windows) with varying position (a), velocity (b) and acceleration (c), and the MAE resulting from the different feedforward controllers (right windows) corresponding to each reference.}
    \label{fig:Robustness_Varying_CLM}
\end{figure}

We demonstrate robustness of the PGNN feedforward controllers by making variations of the reference in Fig.~\ref{fig:CLM_TrackingError1} by changing either the end position ($-0.05$, $0$, $0.05$, $0.1$, $0.15$ $m$), the maximum velocity ($0.025$, $0.05$, $0.075$, $0.1$, $0.125$, $0.15$, $0.2$ $\frac{m}{s}$), and the maximum acceleration ($0.25$, $0.5$, $1$, $2$, $4$ $\frac{m}{s^2}$).
The result is $5+7+5 = 17$ different references, for each of which we apply the physics--based, NN--based, PINN--based and PGNN--based (with $\gamma = 0$ and $\gamma = 0.1$) feedforward controllers, and compute the mean--absolute error (MAE) as
\begin{equation}
    \label{eq:MAE}
    MAE = \frac{1}{N} \sum_{k = 0}^{N-1} | \, e(k) \, |.
\end{equation}
The results are visualized in Fig.~\ref{fig:Robustness_Varying_CLM}, with, in particular, Fig.~\ref{fig:Robustness_VaryingPosition} showing the effect of the end position, Fig.~\ref{fig:Robustness_VaryingVelocity} the effect of the velocity, and Fig.~\ref{fig:Robustness_VaryingAcceleration} the effect of the acceleration.
Fig.~\ref{fig:Robustness_Varying_CLM} demonstrates the robust performance achieved by the PGNN, which, generally outperforms the physics--based feedforward controller with a factor $> 2$ in terms of the MAE for most references.

Similar to Fig.~\ref{fig:CLM_TrackingError2}, we observe the enhanced robustness imposed by the regularization term~\eqref{eq:RegularizationTerm}, i.e., the increase in MAE for the NN--based and PGNN--based ($\gamma = 0$) when the end position is $0.15$ $m$ is not observed for the PINN and PGNN ($\gamma = 0.1$).
In addition, by relying predominantly on the physics--based model, the PGNN demonstrates improved robustness compared to the (PI)NNs feedforward controllers when operated on higher velocities in Fig.~\ref{fig:Robustness_VaryingVelocity}. 
These velocities were not much represented in the training data set $Z^N$, see Fig.~\ref{fig:ExtrapolationDataSet}.

Fig.~\ref{fig:Robustness_VaryingAcceleration} shows a loss in performance of the NN, PINN, and PGNNs for small values of the acceleration. 
This observation can be explained by the flexibility of the NNs in these models, which did not observe any such acceleration during the training process due to the limited accuracy of the encoder measurements.
Namely, reconstructing the acceleration as $\Delta \delta^2 y(k)$ from $y(k)$ that is measured with steps of $5 \cdot 10^{-6}$ $m$, shows that we reconstruct the acceleration in the data set $Z^N$ at discrete values $n \frac{5 \cdot 10^{-6}}{2 (2T_s)^2} = n \frac{5}{8}$ $\frac{m}{s^2}$, $n \in \mathbb{Z}$. 

\label{sec:ValidationSimulation}
\begin{figure}
	\centering
	\includegraphics[width=0.8\linewidth]{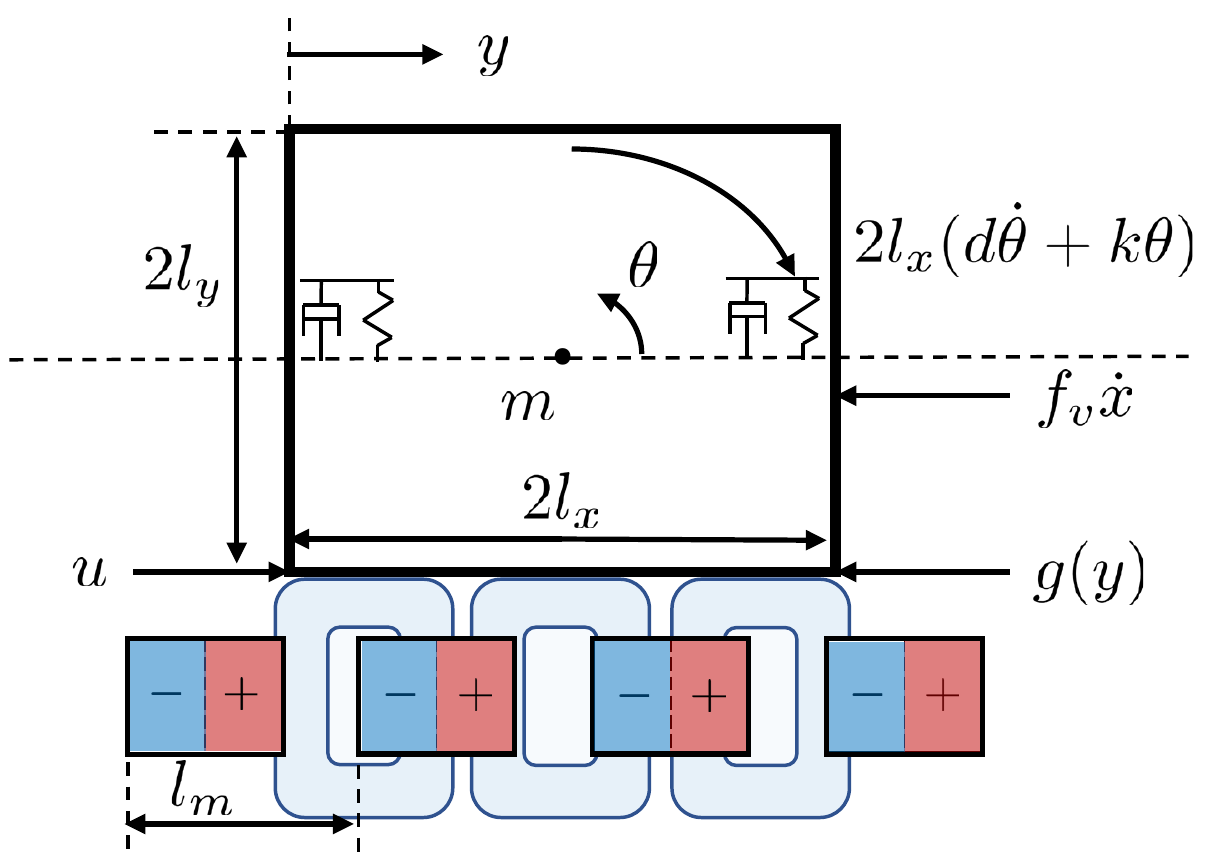}
	\caption{Rotating--translating mass with actuation and sensing on opposite sides of the centre of mass.}
	\label{fig:MotivatingExample}
\end{figure}
\subsection{Nonminimum phase rotating--translating mass}
In the previous section, the discretized dynamics of the CLM~\eqref{eq:CLM_DiscreteTime} yielded $n_b = 1$, such that the feedforward input $u_{\textup{ff}}(k)$ was not a function of past feedforward inputs $u_{\textup{ff}}(k-i)$, $i\in \mathbb{Z}_{>0}$. 
Therefore, the feedforward controllers derived from the models~\eqref{eq:CLM_PhysicsBased}, \eqref{eq:CLM_NNBased} and \eqref{eq:CLM_PGNNBased} are stable by default.
In this section, we demonstrate the efficacy of the PGNN feedforward control methodology by considering the higher order dynamical system illustrated in Fig.~\ref{fig:MotivatingExample}, which is nonminimum phase.

\begin{table}
\caption{Parameter values of the rotating--translating mass displayed in Fig.~\ref{fig:MotivatingExample}.}
\label{tab:ParameterValues}
\centering
\begin{tabular}{|c|c|c|c|c|c|c|c|} \hline
$m$ & $l_x, l_y$ & $M$ & $f_v$ & $k$ & $d$ & $l_m$ & $c$ \\ \hline \hline
$20$ & $1$  & $\frac{40}{3}$  & $50$  & $\frac{25 \cdot 10^{3}}{3}$  & $\frac{575}{3}$ & $0.05$ & $1$ \\ \hline
$kg$ & $m$ & $kgm^2$ & $\frac{kg}{s}$ & $\frac{kg}{s^2}$ & $\frac{kg}{s}$ & $m$ & $\frac{kgm}{s^2}$ \\ \hline
\end{tabular}
\end{table}
\emph{\textbf{System dynamics:}}
we consider a translating--rotating mass with force input $u(k)$ and position output $y(k)$ at opposite sides of the centre of mass, see Fig.~\ref{fig:MotivatingExample}. 
The continuous--time dynamics are given as
\begin{align}
\begin{split}
    \label{eq:DynamicsMotivatingExample}
    M \ddot{\theta}(t) & = l_y \big( u(t)- g \big( y(t) \big) \big) - 2 l_x \big( d \dot{\theta}(t) + k \theta (t)  \big) \\
    m \ddot{x}(t) & = u(t) - f_v \dot{x} - g \big( y(t) \big), \\
    g\big(y(t)\big)  &= c \sin \left( \frac{2 \pi}{l_m} y(t) \right), \\
    y(t) & = x(t) - l_y \theta(t).
\end{split}
\end{align}
In~\eqref{eq:DynamicsMotivatingExample}, $l_x, l_y \in \mathbb{R}_{\geq 0}$ are the width and height of the mass $m \in \mathbb{R}_{>0}$, $M = \frac{1}{3} m (l_x^2+l_y^2)$ is the moment of intertia, $f_v \in \mathbb{R}_{>0}$ the viscous friction, and $d, k \in \mathbb{R}_{>0}$ the damping and spring constant counteracting rotation at both ends of the mass.
The function $g\big(y(t)\big)$ is the cogging force and is assumed \emph{unknown}, with $l_m \in \mathbb{R}_{>0}$ the magnet pole pitch and $c \in \mathbb{R}_{>0}$ the cogging magnitude.
Parameter values are listed in Table~\ref{tab:ParameterValues}.
The system~\eqref{eq:DynamicsMotivatingExample} is controlled in closed--loop by a ZOH--discretized version of
\begin{equation}
\label{eq:FeedbackMotivatingExample}
	C(s) = 5 \cdot 10^3 \frac{s+4 \pi}{s+20 \pi},
\end{equation}
which achieves a $1.22$ $Hz$ bandwidth.

\emph{\textbf{Training data generation:}}
data is generated in closed--loop by sampling $u(k)$ and $y(k)$ at a frequency of $1$ $kHz$. 
We excite the system via the reference $r(k)$ in the top window of Fig.~\ref{fig:TrackingErrorMotivatingExample} for $5$ repetitions, and add a white noise with variance $50$ $N^2$ to the input $u(k)$.

\emph{\textbf{Feedforward controllers:}}
ZOH discretization of the transfer function in~\eqref{eq:DynamicsMotivatingExample} gives $n_a = n_b = 4$, $n_k = 0$.
The identified feedforward controllers are \emph{unstable}, due to the nonminimum phase transfer function in the system~\eqref{eq:DynamicsMotivatingExample}, i.e., $n_{\textup{us}} = 1$. 
We parametrize and train the following feedforward controllers:
\begin{enumerate}
    \item \emph{No feedforward}, i.e., $u_{\textup{ff}}(k) = 0$;
    \item \emph{Physics--based feedforward} with ZPETC stable inversion, i.e.,~\eqref{eq:FeedforwardIdentifiedGeneral},~\eqref{eq:PhysicsBasedParametrization} with linear physical model $f_{\textup{phy}} \big( \theta_{\textup{phy}}, \phi(k) \big) = \theta_{\textup{phy}}^T \phi(k)$ and parameters identified according to~\eqref{eq:IdentificationCriterion} with MSE cost function~\eqref{eq:CostFunctionMSE};
    \item \emph{PGNN--based feedforward} as in~\eqref{eq:PGNNFeedforwardLinear} with $n_1 = 16$ neurons, using ZPETC stable inversion identified according to~\eqref{eq:IdentificationCriterionStable} with cost function~\eqref{eq:CostFunctionPGNN} using $\Lambda_{\textup{NN}} = 0$ and $\Lambda_{\textup{phy}}$ as in~\eqref{eq:RuleOfThumbLambda_phy} with $\epsilon = 1$. Note that $\Theta$ as in~\eqref{eq:SetOptimization} is computed for the stable inverted PGNN;
    \item \emph{Physics--based feedforward} with extended preview, i.e.,~\eqref{eq:PGNNExtendedPreviewWindow} with linear physical model $f_{\textup{phy}} \big( \theta_{\textup{phy}}, \tilde{\phi}(k) \big) = \theta_{\textup{phy}}^T \tilde{\phi}(k)$, $n_{\textup{pw}} = 20$, and parameters identified according to~\eqref{eq:IdentificationCriterion} with MSE cost function~\eqref{eq:CostFunctionMSE};
    \item \emph{PGNN--based feedforward} with linear physical model and $n_{\textup{pw}} =20$ as in~\eqref{eq:PGNNExtendedPreviewWindow} with $n_1 = 16$ neurons identified according to~\eqref{eq:IdentificationCriterionStable} with cost function~\eqref{eq:CostFunctionPGNN} using $\Lambda_{\textup{NN}} = 0$ and $\Lambda_{\textup{phy}}$ as in~\eqref{eq:RuleOfThumbLambda_phy} with $\epsilon = 1$.
\end{enumerate}
We choose $Q = I$ and find $P$ by solving discrete--time Lyapunov equation $\tilde{A}(\theta_{u_{\textup{ff}}}^*)^T P \tilde{A}(\theta_{u_{\textup{ff}}}^*) - P + Q = 0$, where $\tilde{A}(\theta_{u_{\textup{ff}}}^*)$ is obtained either from the ZPETC approximation, or from the linear identification with $n_{\textup{pw}} = 20$. 

\begin{figure}
\centering
	\includegraphics[width=0.8\linewidth]{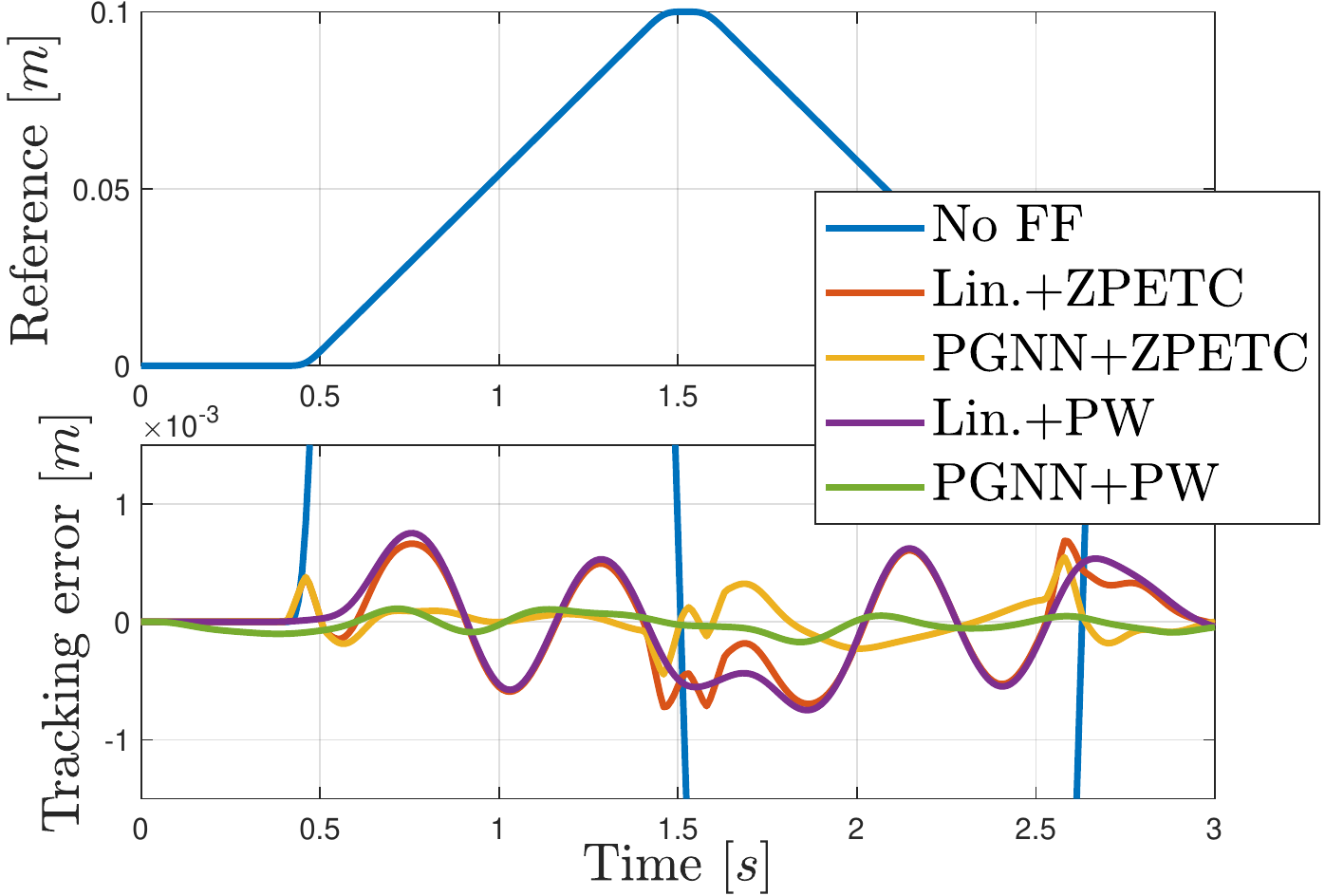}
	\caption{Tracking error resulting from different feedforward controllers.}
	\label{fig:TrackingErrorMotivatingExample}
\end{figure}
\begin{table}
\caption{MSE in $[m^2]$ of the tracking errors in Fig.~\ref{fig:TrackingErrorMotivatingExample}.}
\label{tab:PerformanceMotivatingExample}
\centering
\includegraphics[width=1.0\linewidth]{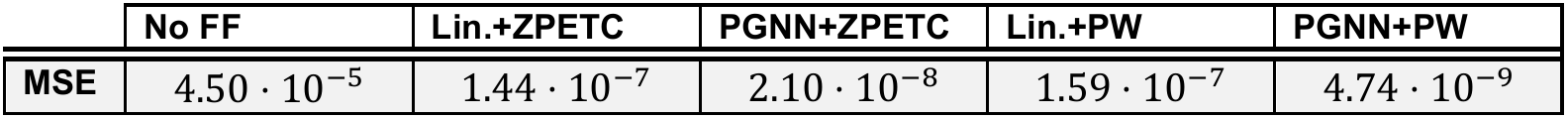}
\end{table}


\emph{\textbf{Results:}}
Fig.~\ref{fig:TrackingErrorMotivatingExample} shows the tracking error resulting from the aforementioned feedforward controllers for the reference shown in the top window. 
It is clear that the PGNN manages to significantly outperform the physics--based feedforward controller in terms of tracking error, as is confirmed by the mean--squared error (MSE) values listed in Table~\ref{tab:PerformanceMotivatingExample}. 

The PGNN improves over the linear, physics--based feedforward controller, since it is capable to identify the nonlinear dynamics $g \big( y(t) \big)$ in~\eqref{eq:DynamicsMotivatingExample}.
In addition, the stable approximation method ZPETC induces some tracking error that is particularly seen during the start of the movement for both the linear and PGNN feedforward controllers.
This error is not present when extending the preview window of the feedforward controllers, which allows for preactuation to deal with the nonminimum phase behaviour, rather than approximating it as does ZPETC.


\section{Conclusions}
\label{sec:Conclusions}
In this paper we developed a novel physics--guided neural network (PGNN) architecture that structurally merges a physics--based layer and a black--box neural layer in a single model. The parameters of the two layers are simultaneously identified, while a novel regularization cost function was used to prevent competition among layers and to preserve consistency of the physics--based parameters. Moreover, in order to ensure stability of PGNN feedforward controllers, we developed sufficient conditions for analyzing or imposing (during training) input--to--state stability of PGNNs, based on novel, less conservative Lipschitz bounds for neural networks. We showed that the developed PGNN feedforward control framework reaches a factor $2$ improvement with respect to physics--based mass--friction feedforward and it significantly outperforms alternative neural network based feedforward controllers for a real--life industrial linear motor and for a challenging non--minimum phase mechatronics example.


\appendix

\section{Proof of Proposition~\ref{prop:ConsistentIdentification}}
\label{app:propConsistentIdentification}
\begin{proof}
    Substitution of the inverse dynamics~\eqref{eq:UnknownDynamics} and PGNN~\eqref{eq:PGNNGeneral} into the cost function~\eqref{eq:CostFunctionPGNN} with $\Lambda_{\textup{NN}} = 0$~gives
    \begin{align}
    \begin{split}
    \label{eq:Proof1Step1}
        \frac{1}{N} & \sum_{i = 0}^{N-1} \Big( f_{\textup{phy}} ( \theta_{\textup{phy}}^*, \phi_i ) + g ( \phi_i ) - f_{\textup{phy}} ( \theta_{\textup{phy}}, \phi_i ) \\
        & - f_{\textup{NN}} \big( \theta_{\textup{NN}}, T(\phi_i) \big) \Big)^2 + \left\| \begin{bmatrix} 0 & 0 \\ 0 & \Lambda_{\textup{phy}} \end{bmatrix} \left( \theta - \begin{bmatrix} 0 \\ \theta_{\textup{phy}}^* \end{bmatrix} \right) \right\|_2^2.
    \end{split}
    \end{align}
    Both terms in~\eqref{eq:Proof1Step1} are non--negative, such that the global minimum is attained for $\hat{\theta}_{\textup{phy}} = \theta_{\textup{phy}}^*$ (regularization term) and $\hat{\theta}_{\textup{NN}} = \theta_{\textup{NN}}^*$ (MSE term, after substitution of $\theta_{\textup{phy}} = \theta_{\textup{phy}}^*$).  
\end{proof}

\section{Proof of Proposition~\ref{prop:OptimizedInitialization}}
\label{app:PropOptimizedInitialization}
\begin{proof}
    The proof of~\eqref{eq:OptimizedInitializationResult1} follows directly by observing that $\theta_{\textup{L}}^{(j)}$ in~\eqref{eq:OptimizedInitialization} is the least squares solution of~\eqref{eq:IdentificationCriterion},~\eqref{eq:CostFunctionPGNN} for the PGNN~\eqref{eq:PGNNLIP} given $\theta_{\textup{NL}}^{(j)}$.
    Since $M(\theta_{\textup{NL}}^{(j)})$ is non--singular, $\theta_{\textup{L}}^{(j)}$ is unique, such that~\eqref{eq:OptimizedInitializationResult1} holds with strict inequality if and only if $\theta_{\textup{L}}^{(j)} \neq \overline{\theta}_{\textup{L}}$. 
    Observe that, $\theta_{\textup{L}}^{(j)} - \overline{\theta}_{\textup{L}} =\theta_{\textup{L}}^{(j)} - M \big( \theta_{\textup{NL}}^{(j)} \big)^{-1} M \big( \theta_{\textup{NL}}^{(j)} \big) \overline{\theta}_{\textup{L}} \neq 0$ gives condition~\eqref{eq:OptimizedInitializationCondition} after substitution of~\eqref{eq:OptimizedInitialization} and~\eqref{eq:OptimizedInitializationMatrix} and using nonsingularity of~$M(\theta_{\textup{NL}}^{(j)})$. 

    Secondly,~\eqref{eq:ImprovedMSE} follows by observing that choosing $\theta_{\textup{L}}^{(j)} \neq \overline{\theta}_{\textup{L}}$ cannot decrease $V_{\textup{reg}}$ when $\Lambda =0$. 
    Correspondingly,~\eqref{eq:OptimizedInitializationResult1} states that $V_{\textup{MSE}}$ must decrease if~\eqref{eq:OptimizedInitializationCondition} holds, such that
    \begin{equation}
        \label{eq:Lemma41_Step1}
        V_{\textup{MSE}} \left( \begin{bmatrix} \theta_{\textup{NL}}^{(j)} \\ \theta_{\textup{L}}^{(j)} \end{bmatrix}, Z^N \right) \leq V \left( \begin{bmatrix} \theta_{\textup{NL}}^{(j)} \\ \overline{\theta}_{\textup{L}} \end{bmatrix}, Z^N \right) = V_{\textup{MSE}} (\theta_{\textup{phy}}^*, Z^N ),
    \end{equation}
    and with strict inequality if~\eqref{eq:OptimizedInitializationCondition} holds.
\end{proof}

\section{Proof of Theorem~\ref{th:StabilityPGNN}}
\label{app:ThPGNNStability}
\begin{proof}
    The proof follows by showing that $ V\big( x(k) \big) = x(k)^T P x(k)$ is an ISS--Lyapunov function as in Definition~\ref{def:ISSLyapunov}. 
    Condition~\eqref{eq:ISSLyapunov1} is satisfied with $\kappa_1 \big(\|\phi_{u_{\textup{ff}}}(k) \| \big) = \lambda_{\textup{min}}(P) \phi_{u_{\textup{ff}}}(k)^T \phi_{u_{\textup{ff}}}(k)$ and $\kappa_2 \big( \|\phi_{u_{\textup{ff}}}(k)\| \big) = \lambda_{\textup{max}}(P) \phi_{u_{\textup{ff}}}(k)^T \phi_{u_{\textup{ff}}}(k)$. 
    We compute the difference $V \big( \phi_{u_{\textup{ff}}}(k+1) \big) - V \big( \phi_{u_{\textup{ff}}}(k) \big)$ to obtain
    \begin{align}
    \begin{split}
        \label{eq:Proof2Step1}
        &V \big( \phi_{u_{\textup{ff}}}(k+1) \big) - V \big( \phi_{u_{\textup{ff}}}(k) \big) = - \phi_{u_{\textup{ff}}}(k)^T Q \phi_{u_{\textup{ff}}}(k) \\
        & \quad \quad \quad \quad  + 2 C_1 B^T P A \phi_{u_{\textup{ff}}}(k) +  B^T P B C_1^2, \\
        &C_1 := \hat{\theta}_r^T \phi_{{r}}(k+1) + f_{\textup{NN}} \left( \hat{\theta}_{\textup{NN}}, \begin{bmatrix} \phi_{r}(k+1) \\ \phi_{u_{\textup{ff}}}(k) \end{bmatrix} \right).
    \end{split}
    \end{align}
    For a term $2p^T q$ and a $\varepsilon \in \mathbb{R}_{>0}$, we can complete the squares as:
    \begin{equation}
        \label{eq:CompletingTheSquares}
        2p^T q = \varepsilon p^T p - \varepsilon (p - \frac{1}{\varepsilon} q)^T (p - \frac{1}{\varepsilon} q) + \frac{1}{\varepsilon} q^T q \leq \varepsilon p^T p + \frac{1}{\varepsilon} q^T q. 
    \end{equation}
    We complete the squares~\eqref{eq:CompletingTheSquares} of $2 C_1 B^T P A \phi_{u_{\textup{ff}}}(k)$ in~\eqref{eq:Proof2Step1} using $\varepsilon = \beta \lambda_{\textup{min}}(Q)$, $\beta \in \mathbb{R}_{>0}$, $p = \phi_{u_{\textup{ff}}}(k)$, and $q = A^T P B C_1$ to obtain
    \begin{align}
    \begin{split}
        \label{eq:Proof2Step2}
        V \big( & \phi_{u_{\textup{ff}}}(k+1) \big) - V \big( \phi_{u_{\textup{ff}}}(k) \big)  \leq - (1 - \beta ) \lambda_{\textup{min}} (Q) \phi_{u_{\textup{ff}}}(k)^T \phi_{u_{\textup{ff}}}(k) + c_{\beta} C_1^2, 
    \end{split}
    \end{align}
    with $c_{\beta} := B^T P B \left( I + \frac{1}{\beta \lambda_{\textup{min}}(Q)} A A^T P \right) B$. 
    Similarly, by substituting $C_1$ in~\eqref{eq:Proof2Step2} and completing the squares~\eqref{eq:CompletingTheSquares} for $2 \hat{\theta}_r^T \phi_r(k) f_{\textup{NN}}$ using $\varepsilon = \beta_1$, $\beta_1 \in \mathbb{R}_{>0}$, $p = f_{\textup{NN}}$ and $q = \hat{\theta}_r^T \phi_r(k+1)$, we obtain
    \begin{align}
    \begin{split}
        \label{eq:Proof2Step3}
        &V \big( \phi_{u_{\textup{ff}}}(k+1) \big) - V \big( \phi_{u_{\textup{ff}}}(k) \big) \leq - (1-\beta) \lambda_{\textup{min}}(Q) \phi_{u_{\textup{ff}}}(k)^T \phi_{u_{\textup{ff}}}(k) +\\
        & c_{\beta} ( 1 + \beta_1 ) f_{\textup{NN}} \left( \hat{\theta}_{\textup{NN}}, \begin{bmatrix} \phi_{{r}}(k+1) \\ \phi_{u_{\textup{ff}}}(k) \end{bmatrix} \right)^2 + c_{\beta} ( 1 + \frac{1}{\beta_1} ) \big(\hat{\theta}_r^T \phi_{{r}}(k+1) \big)^2,
    \end{split}
    \end{align}
    with $\beta_1 \in \mathbb{R}_{>0}$. 
    Finally, substitution of the Lipschitz condition~\eqref{eq:NNLipschitz} with $\phi_{\textup{ff}}^B (k) = 0$ and~\eqref{eq:EquilibriumOrigin}, and completing the squares~\eqref{eq:CompletingTheSquares} for $2 K_{u_{\textup{ff}}}^T \phi_{u_{\textup{ff}}}(k) K_{r}^T \phi_r(k+1)$ using $\varepsilon = \beta_2$, $\beta_2 \in \mathbb{R}_{>0}$, $p = K_{u_{\textup{ff}}}^T \phi_{u_{\textup{ff}}}(k)$, and $q = K_{r}^T \phi_r(k+1)$, gives
    \begin{align}
    \begin{split}
        \label{eq:Proof2Step4}
        &V \big( \phi_{u_{\textup{ff}}}(k+1) \big) - V \big( \phi_{u_{\textup{ff}}}(k) \big) \leq - \phi_{u_{\textup{ff}}}(k)^T \Big( (1-\beta) \lambda_{\textup{min}}(Q) - c_{\beta} \cdot \\
        & \quad \quad (1+\beta_1) (1+\beta_2) K_{u_{\textup{ff}}} K_{u_{\textup{ff}}}^T \Big) \phi_{u_{\textup{ff}}}(k) + \phi_r (k+1)^T c_{\beta} \cdot \\
        & \quad \quad \Big( (1+\beta_1) (1 + \frac{1}{\beta_2}) K_r K_r^T + (1 + \frac{1}{\beta_1}) \hat{\theta}_r \hat{\theta}_r^T \Big) \phi_r(k+1) \\
        & =: - \kappa_3 \big(\| \phi_{u_{\textup{ff}}}(k)\| \big) + \sigma \big( \| \phi_r(k+1) \| \big).
    \end{split}
    \end{align}
    It is clear that $\kappa_3 \big( | \phi_{u_{\textup{ff}}}(k) | \big)$ is a $\mathcal{K}_{\infty}$--function if the matrix between $\phi_{u_{\textup{ff}}}(k)^T$ and $\phi_{u_{\textup{ff}}}(k)$ is positive definite, which reduces to the scalar condition~\eqref{eq:StabilityPGNNCondition} when using the Cauchy--Schwartz inequality, i.e., substitute $\big( K_{u_{\textup{ff}}}^T \phi_{u_{\textup{ff}}}(k) \big)^2 = \phi_{u_{\textup{ff}}}(k)^T K_{u_{\textup{ff}}} K_{u_{\textup{ff}}}^T \phi_{u_{\textup{ff}}}(k) \leq K_{u_{\textup{ff}}}^T K_{u_{\textup{ff}}} \phi_{u_{\textup{ff}}}(k)^T \phi_{u_{\textup{ff}}}(k)$ and choose $\beta_1$, $\beta_2$ arbitrary small.
\end{proof}

\section{Proof of Lemma~\ref{le:StabilityAPrioriWithAssumption}}
\label{app:leStability1}
\begin{proof}
    The proof follows directly from Proposition~\ref{th:StabilityPGNN} and observing that limiting $\theta \in \Theta$ ensures that~\eqref{eq:StabilityPGNNCondition} is satisfied by using the NN Lipschitz bound in~\eqref{eq:NNLipschitzValue}. 
\end{proof}

\bibliographystyle{elsarticle-num} 
\bibliography{References}




\end{document}